\documentclass[dvinames, 11pt]{article}

\usepackage{booktabs}   
\usepackage{subcaption} 
\usepackage[left=3cm, right=3cm, bottom=2.4cm, top=2.4cm]{geometry}
\usepackage{lscape}
\usepackage{amsmath}
\usepackage{amsthm}
\usepackage{amssymb}
\usepackage{latexsym}

\usepackage{proof}
\usepackage{bussproofs}
\usepackage{amssymb}
\usepackage{amsmath}
\usepackage{mathrsfs}

\newtheorem{theorem}{Theorem}[section]
\newtheorem{lemma}[theorem]{Lemma}
\newtheorem*{theorem*}{Theorem}
\newtheorem{proposition}[theorem]{Proposition}

\theoremstyle{definition}
\newtheorem{definition}{Definition}[section]
\newtheorem{example}{Example}[section]
\theoremstyle{remark}
\newtheorem{remark}{Remark}[section]

\newcommand{\ergo}{\Rightarrow}
\newcommand{\impl}{\rightarrow}

\newcommand{\fal}{\bot}

\usepackage{cmll}
\newcommand{\pa}{\parr}
\newcommand{\te}{\otimes}
\newcommand{\lo}{\multimap}

\newcommand{\lam}{\lambda}
\newcommand{\send}[1]{\overline{#1}.\,}

\newcommand{\lammu} {\lambda \mu}
\newcommand{\lamp} {\lambda_\pa}
\newcommand{\p}{\mid}

\newcommand{\distp}[2]{\overline{#1}\, \overline{#2}\,}
\newcommand{\sendm}[3]{\overline{#1} #3.#2}

\newcommand{\termfal}{\circ}

\newcommand{\calc}{\mathcal{C}}
\newcommand{\cald}{\mathcal{D}}
\newcommand{\calp}{\mathcal{P}}
\newcommand{\cale}{\mathcal{E}}

\newcommand{\I}{\mathrm{I}}
\newcommand{\E}{\mathrm{E}}

\newcommand{\cut}{\text{cut}}

\newcommand{\red}{\mapsto}

\newcommand{\NLL}{\mathrm{NMLL}}
\newcommand{\SLL}{\mathrm{MLL}}

\newcommand{\cs}{\mathrm{cs}}

\newcommand{\abort}{{\mathsf{close}}}

\title{$\pa$ means Parallel:  Multiplicative Linear Logic Proofs as Concurrent Functional Programs}
\author{
  \begin{tabular}{ccc} Federico Aschieri\footnote{Funded by FWF grant
P32080-N31.} &$\qquad\qquad$& Francesco A. Genco\footnote{Funded by
ANR JCJC project Intuitions Bolzaniennes.} \\TU Wien & & IHPST,
Universit\'e Paris 1
  \end{tabular}}
\begin{document}
\maketitle

\begin{abstract}
Along the lines of the Abramsky ``Proofs-as-Processes'' program, we present an  interpretation of multiplicative linear logic as typing system for concurrent functional programming. In particular, we study a 
linear multiple-conclusion natural deduction system and show it is isomorphic to a simple and natural extension of $\lambda$-calculus with parallelism and communication primitives, called $\lamp$. We shall prove that $\lamp$ satisfies all the desirable properties for a typed programming language: subject reduction, progress, strong normalization and confluence.
\end{abstract}

\section{Introduction}

\subsection*{The Proofs-as-Processes Program}
Introduced by Girard in 1987, linear logic was announced right off the bat as the logic of parallelism. According to  \cite{Gir87}, the ``\emph{connectives of linear logic have an obvious meaning in terms of parallel computation (...) In particular, the multiplicative fragment can be seen as a system of communication without problems of synchronization}.'' 
Stimulated by this remark and the further observation that ``\emph{cut elimination is parallel communication between processes}'' (\cite{GLT89}, pp. 155), Abramsky \cite{Abramsky94} launched the ``Proofs-as-Processes'' program, whose goal was to support  with computational evidence Girard's claims. Namely, the goal was  ``\emph{to show how a process calculus, sufficiently expressive to allow a reasonable range of concurrent programming examples to be handled, could be exhibited as the computational correlate of a proof system}'' \cite{Abramsky94}.  A candidate process calculus was Milner's $\pi$-calculus \cite{Milner}; a possible  proof system, the linear sequent calculus. The hope was that linear logic could provide a canonical and firm foundation to concurrent computation, serving as a tool for typing and reasoning about communicating processes. 

\subsection*{Early successes}

  After the early work of \cite{BellinScott},  the turning point was the discovery by Caires and Pfenning \cite{CairesPfenning} of a tight correspondence between 
  intuitionistic linear logic and a session typed $\pi$-calculus, the motto being: linear 
  propositions as session types, proofs as processes, cut-elimination as communication. This was yet another instance of the famous Curry-Howard 
  correspondence \cite{Wad15}, linking proofs to programs, propositions to types, proof normalization to computation. Soon after, Wadler \cite{Wadler2012} introduced the session typed $\pi$-calculus CP, shown to tightly correspond to classical linear logic. 
  
  \subsection*{Limitations}

This important research notwithstanding, it appears that there is still ground to cover toward canonical and firm foundations for concurrent computation. First of all, asynchronous communication does not appear to rest on solid foundations. Yet this communication style is easier to implement and more practical than the purely synchronous paradigm, so it  is   widespread  and asynchronous typed process calculi have been already investigated (see~\cite{Gay2010}). Unfortunately, linear logic has not so far provided via Curry-Howard a logical account of asynchronous communication. The reason is 
 that there is a glaring discrepancy between \emph{full} cut-elimination and $\pi$-calculus reduction which has not so far been addressed. On one hand, as we shall see, linear logic does support asynchronous communication, but only through the \emph{full} process of cut-elimination, which indeed makes essential use of asynchronous communication. On the other hand, the $\pi$-calculus of \cite{CairesPfenning} only mimics a \emph{partial} cut-elimination process that only eliminates top-level cuts. Indeed, by comparison with its type system, the $\pi$-calculus lacks some necessary computational reduction rules.
Some of the missing  reductions, corresponding to commuting conversions, were provided in Wadler's CP. The congruence rules that allow the extra reductions to mirror full cut-elimination, however, were rejected: in Wadler's \cite{Wadler2012} own words, ``\emph{such rules do not correspond well to our notion of computation on processes, so we omit them}''. A set of reductions similar to that rejected by Wadler is regarded in \cite{Perez2012} as sound relatively to a notion of ``typed context bisimilarity''. The notion, however, only ensures that two ``bisimilar'' processes have the same input/output behaviour along their main communication channel; the internal synchronization among the parallel components of the two processes is not captured and may differ significantly. Thus, the extensional flavor of the bisimilarity notion prevents it to detect the intensionally different behaviour of the related processes, that is, how differently they communicate and compute.

There is no way out: attempts to model asynchronous communication by mirroring full cut-elimination with Milner's $\pi$-calculus are bound to fail. 
An example will clarify the matter. In the process $\nu z( x(y).z(w). P \p z[a].Q)$, the second process $z[a].Q$ wishes to transmit the message $a$ along the channel $z$ to the first process. 
However, the first process is not ready to receive on $z$ yet, because it is designed to first receive a message on the channel $x$. In Milner's $\pi$-calculus, this is a deadlock, because input and output actions are rigidly ordered and \emph{blocking}, for the sake of synchronization. Indeed, as the heir of
CCS~\cite{Mil80}, $\pi$-calculus is in its essence a formalism for
representing synchronization: a communication happens when two
processes synchronize by consuming two dual constructs, as for
instance $\overline{x}\, m$ and $x(y)$. However, when a process is seen as proof in the logical system, the order between prefixes becomes less relevant, thus in CP there is an extra commuting conversion $$\nu z( x(y).z(w). P \p z[a].Q)\mapsto x(y). \nu z( z(w).P\p z[a].Q)$$
In CP, there is no further reduction allowed, and rightly so: otherwise the blocking nature of the prefix $x(y)$ would be violated, making the very construct pointless for synchronization. On the other hand, for cut-elimination the now possible reduction
$$x(y). \nu z( z(w).P\p z[a].Q) \mapsto x(y). \nu z( P[a/w]\p Q)$$
is needed. In this case, however, as we have seen the very existence of a construct for input becomes misleading and would be more coherent to drop it altogether, so that the first process would be directly re-written and reduced as follows
$$\nu z( P \p w[a].Q)\mapsto \nu z( P[a/w]\p Q)$$
 This is indeed the approach of the present paper. Sensitive to similar concerns, \cite{KMP2019} introduced HCP, which features delayed input-output actions, modeling a form of asynchrony. However, the change in the transitions with respect to CP undermines the synchronous nature of Milner's $\pi$-calculus, as we have seen, making it quite a different calculus and jeopardizing its canonicality ambitions. Moreover, HCP comes at the cost of adding redundant logical structures to the linear sequent calculus, namely hypersequents \cite{Avron91}, which were instead used by Avron to \emph{increase} the logical power of sequents. 

There are other limitations in the current linear logic foundations of concurrent computation.  The linear sequent calculus fails to  combine seamlessly functional and concurrent computation.  Indeed, according to Abramsky et al. \cite{Abramsky96} it is a ``\emph{major open problem (...) to combine our understanding of the functional and concurrent paradigms, with their associated mathematical underpinnings, in a single unified theory}''. The standard sequent calculus falls short in several other respects: it prevents deadlocks, but by excessively restricting possible communication patterns, for instance by forbidding cyclic configurations; it does not permit in any simple way code mobility, that is, the ability to communicate code, like in higher-order $\pi$-calculus \cite{San1993}; it does not represent multi-party communication sessions; the syntax of processes does not match exactly the syntax of the typing system. 
 In order to address some of these pitfalls,  much work has been done, by either considering non-standard proof systems for linear logic  or by  extending the logic itself; unfortunately for the canonical-foundation enterprise and practical usability, these systems tends to be proof-theoretically artificial. A system adding a functional language on top of linear-logic-based typed $\pi$-calculus was studied in \cite{TCP2013}.   A solution for allowing, also in the context of the session typed $\pi$-calculus, cyclic communication patterns,  was given in  \cite{Dardha} by coming up with a new variant of linear logic; this makes possible for example to implement Milner's scheduler process. Multi-party session types  has been treated for example in \cite{MontesiCarbone}. 
 A hypersequent proof system for linear logic whose syntax perfectly matches the syntax of $\pi$-calculus has been introduced in \cite{KMP2019}. A logic for typing asynchronous communication in a non-standard process calculus has been proposed in \cite{Pfenning2018}.

\subsection*{$\lamp$: A concurrent extension of $\lambda$-calculus arising from linear logic}

We shall illustrate a new approach in the study of linear logic  as typing system for concurrent programs. 
Our framework does not suffer any of the mentioned limitations of linear sequent calculus.
 Inspired by \cite{paivalinear, Par91}, we move away from sequent calculus and adopt a multiple-conclusion natural deduction for multiplicative linear logic, which we show to be isomorphic to a concurrent extension of $\lambda$-calculus, called  $\lamp$. As a result, $\lamp$ is typed by linear multiple-conclusion natural deduction; it is naturally asynchronous,  yet it can model synchronous communication as well; it allows arbitrary communication patterns, like cyclic ones, but prevents deadlocks; it allows higher-order communication and code mobility;  it mirrors perfectly a full normalization procedure resulting in analytic proofs; its types can be read as program specifications in the traditional way, although they do not describe communication protocols as session types do.

The typing system of $\lamp$, in its non-linear variant, has been very well
studied: computationally, it was famously interpreted by Parigot
\cite{Parigot92} as $\lammu$-calculus; proof-theoretically, it was thoroughly investigated by \cite{CELLUCCI1992}.
Although proof-nets have sometimes been dubbed ``\emph{the natural deduction of linear logic}'', our typing system is closer to a natural deduction.
 It is not exactly  \emph{natural} deduction \cite{Prawitz}, since it is multiple-conclusion, hence it is natural for building and typing parallel programs, but  not  much so for modeling  \emph{human} deduction.
   However, our system is based on natural-deduction normalization rather than sequent-calculus cut-elimination and linear implication $\lo$ is a primitive connective, with the standard introduction and elimination rules. 
Like proof-nets, $\lamp$ abstracts away from those inessential permutations in the order of rules that plague multiple-conclusion logical systems.
 As a result, $\lamp$ avoids commuting conversions, which have never been convincing from the computational point of view. Actually, $\lamp$ is not only a concurrent $\lambda$-calculus, it also looks like the natural deduction version of proof nets. Finally, it enjoys all the good properties that a well-behaved functional programming language should have: subject reduction, progress, strong normalization, and confluence. It is a step forward in the direction of that elusive concurrent $\lambda$-calculus which Milner attempted to find, before creating CCS out of the failure\footnote{According to Milner \cite{Milner84}, ``\emph{CCS is an attempt to provide a paradigm for concurrent computation. It arose after several unsuccessful attempts by the author to finds a satisfactory generalization of the lambda calculus, to admit concurrent computation.}''}.

\subsection*{The Insight}

The main insight behind the typing system of $\lamp$ is to break into two rules the linear implication introduction. One rule is intuitionistic and yields  functional abstraction, the second is classical and yields communication. Namely, the first rule, taken from \cite{paivalinear}, introduces the $\lambda$-operator, the second a communication channel with continuation:
  \[\infer[\lo \I\ (\mbox{ if $x$ occurs in $t$})]{    \Gamma \ergo \lambda{x}\,  t: A \lo B , \Delta }{   
\Gamma , x: A \ergo t: B  , \Delta  }\]

\[\infer[\lo \I\ (\mbox{ if $x$ occurs  in $\Delta$})] {    \Gamma \ergo \send{x}  t: A \lo B , \Delta }{   
\Gamma , x: A \ergo t: B  , \Delta  }
\]
Indeed, a communication channel   $\send{x}  t$ is supposed to take as argument a term $u: A$, transmit it along the channel $x$ and then continue with $t: B$. We shall indeed denote the application of  $\send{x}  t$ to $u$ as $\sendm{x}{t}{u}$, as in $\pi$-calculus: a promising sign!

The parallel operator $\p$ is introduced by the $\pa$ introduction rule:
\[ \infer[\pa \I ]{ \Gamma \ergo s\p t : A\pa B , \Delta }{ \Gamma
\ergo s: A , t: B , \Delta } \]
finally doing justice to the $\pa$ connective, called \emph{par} since the beginning.


\section{Multiplicative Linear Logic  as a Concurrent Functional  Calculus } \label{sec:calculus}

In this section, we present the syntax, the typing system $\NLL$ and the reduction rules of the concurrent functional calculus $\lamp$. 


\subsection{The calculus $\lamp$}

\begin{definition}[Terms of $\lamp$]\label{def:typing} The
\emph{untyped terms of $\lamp$} are defined by the following grammar:
$$\begin{aligned} u, v::=\ &\mbox{ }\termfal\qquad &(\mbox{terminated process})\\
&|\  \send{x}{u} \qquad &(\mbox{output channel $\overline{x}$ with continuation})\\ 
&|\ \distp{x}{y}\, u\qquad &(\mbox{output channels $\overline{x}, \overline{y}$ without continuation})\\
&|\  u \p v  \qquad &(\mbox{parallel composition})\\
&|\ \abort (u) \qquad &(\mbox{process termination})\\
&|\ x \qquad &(\mbox{variable})\\
&|\  u\, v\qquad &(\mbox{function application}) \\ 
& \mathsf{syntactic}\ \mathsf{sugar}\\
&|\ \lambda x\, u := \send{x}{u}\qquad &\mbox{ (function definition (provided $x$  occurs in $u$))}\\
&|\ \sendm{x}{u}{v} := (\send{x}{u})v 
\end{aligned}$$
\end{definition} The calculus $\lamp$ is nothing but a linear lambda
calculus extended with communication primitives. When $x$ occurs in
$u$, the operator $\send{x}{u}$ is denoted and behaves exactly as
$\lambda x\, u$. When $x$ does not occur in $u$, the term
$\send{x}{u}$ behaves as an output channel $\overline{x}$ followed by
a continuation $u$. Channels in $\lamp$ are thus first-class citizens
that can be transmitted, shared, moved and applied to messages just
like functions to arguments.  When the channel $\send{x}{u}$ is
applied to an argument $v$, it will be denoted as $\sendm{x}{u}{v}$,
because it behaves exactly as its $\pi$-calculus counterpart: it
transmits the message $v$ and continues with $u$. The term
$\distp{x}{y}\, u$ is an output channel which transmits the two
messages contained in $u$ and then terminates.

If we omit parentheses, each term of $\lamp$ can be rewritten, not
uniquely, as a parallel composition of several terms: $t_{1}\p t_{2}\p
\ldots \p t_{n}$.  We will adopt this notation throughout the paper
since it is irrelevant how the parallel operator $|$ associates.

\begin{remark} An alternative choice for the syntax of $\lamp$ may
have been:
$$ u, v\quad ::= \quad x \quad  |\quad  u\,v \quad |\quad  \lambda x\, u\quad | \quad \sendm{x}{u}{v} \quad \mid\quad u \p v\quad |\quad \distp{x}{y}\, u\quad |\quad \termfal \quad |\quad\abort (u)$$
Then the operator $\sendm{x}{u}{v}$ would have been typed by a
cut-rule and the construct $\send{x}u$ would have been defined as:
$\lambda y\, \sendm{x}{u}{y}$. This approach, however, would have
raised some technical complications: when transmitting a message $v$,
one would have to deal with the free variables of $v$ which are bound
by $\lambda$ operators occurring outside $v$. A similar issue is
solved in~\cite{ACG2018gandalf}, and a simpler solution could be
adopted here. Nevertheless, we adopt the present approach as it is the
smoothest.
\end{remark}

\subsection{The Typing System $\NLL$}
\begin{definition}[Types]
The \textbf{types} of $\NLL$ are built in the usual way from propositional variables, $\fal$ and the connectives $\lo$ and $\pa$.
\end{definition}
\noindent Since the typing system of $\lamp$ is classical linear
logic, the $\te$ connective can be defined in terms of the
others. Indeed, its computational interpretation seems less primitive,
as it does not define any computational construct that is not already
captured by $\lo$ and $\pa$. Hence, we leave it out of the
system.

\begin{definition}[Sequents]\label{def:sequent}
The judgments of the typing system $\NLL$ for $\lamp$ are \textbf{sequents} $\Gamma \ergo \Delta$, where:

1. $\Gamma=x_{1}: A_{1}, \ldots, x_{n}: A_{n}$ is a sequence of distinct variable declarations, each $x_{i}$ being a variable, $A_{i}$ a type and $x_{i}\neq x_{j}$ for $i\neq j$.

2. $\Delta= t_{1}: B_{1}, \ldots, t_{n}: B_{n}, t_{n+1}, \ldots, t_{m}$, each $t_{i}$ being a term of $\lamp$ and $B_{i}$ a type. 

\end{definition}
\noindent The typing system of $\lamp$ is a linear version of the
multiple-conclusion classical natural deduction
of~\cite{CELLUCCI1992}. It is presented in
Table~\ref{tab:type-assignment} and discussed in detail below.  As
usual, we interpret a derivable sequent $\Gamma \ergo \Delta$ as a
judgement giving types to a sequence of terms $\Delta$, provided their
variables are used with types declared in $\Gamma$. The terms in
$\Delta$ that have no type are processes that return no value, and
always have the form $\abort(u)$, with $u: \fal$. These terms however
cannot be ignored, as they may be involved in some communication; thus
they are evaluated and then terminated.

\subsection{Linearity of Variables}
 
The only condition that the typing rules must satisfy in order to be
applied is that their premises share no variable. This amounts to
requiring the linearity of variables and to making sure that  there is no clash
among the names of the output channels. As we shall see in
Proposition~\ref{prop:linchan}, the result of this implicit renaming
is that a typed term cannot contain two distinct occurrences of the
same output channel $\overline{x}$ nor of the same input channel
$x$. This property has two consequences. First,  in $\lamp$ we do not need, and
actually do not use, any renaming during reductions.
Secondly, communication channels cannot be used twice: once a message
is transmitted along a channel, the channel is consumed and disappears
forever, both  from the sender and from  the receiver of the message. These two
processes can of course keep communicating with each other or with
other processes, but they must use different channels for further
communications.


\begin{table}[h!]
  \centering
  \begin{center}
  \hrule
\begin{small}

\[ x:A \ergo x:A\]

  \[\infer[\lo \I\ (\mbox{ if $x$ occurs in $t$})]{    \Gamma \ergo \lambda{x}\,  t: A \lo B , \Delta }{   
\Gamma , x: A \ergo t: B  , \Delta  }\]

\[\infer[\lo \I\ (\mbox{ if $x$ occurs  in $\Delta$})] {    \Gamma \ergo \send{x}  t: A \lo B , \Delta }{   
\Gamma , x: A \ergo t: B  , \Delta  }
\]


\[\infer[\lo \E]{ \Gamma , \Sigma  \ergo st:  B  , \Delta , \Theta }{     \Gamma \ergo s: A
    \lo B , \Delta  &&     \Sigma \ergo t: A , \Theta }\]
   
%

\[ \infer[\pa \I ]{ \Gamma \ergo s\p t : A\pa B , \Delta }{ \Gamma
\ergo s: A , t: B , \Delta } \]


\[ \infer[\pa \E ]{\Gamma , \Sigma_1 ,\Sigma _2 \ergo \distp{x}{y}  s:\fal , \Delta , \Theta _1 ,\Theta _2
}{ \Gamma \ergo s: A \pa B , \Delta && \Sigma_1 , x: A \ergo
\Theta _1 && \Sigma _2 , y:B \ergo \Theta _2 }
\]

\[\infer[\fal\I]{\Gamma \ergo \Delta , \termfal :\fal }{\Gamma \ergo
\Delta}
\qquad \infer[\fal \E ]{\Gamma \ergo \abort(t), \Delta}{\Gamma \ergo  t: \fal , \Delta}
\]

\medskip 

Each rule in this table can only be applied if its premises share no variable 
\end{small}
\medskip

\hrule
\end{center}
\caption{The type assignment rules of $\NLL$.}
\label{tab:type-assignment}
\end{table}

\subsection{$\pa$ as a Parallel Operator}
 The typing system of $\lamp$ is tailored to parallel computation. It is based  on the view that the different processes of a parallel program are 
 not independent entities, but make sense only as components of a larger and more complex system: they share resources through communication and they 
 exploit the computational power of one another. 
Hence, one cannot build and type one process, without building at the same time all the interconnected processes.
Coherently, a derivation of a sequent $\Gamma \ergo t_{1}: A_{1}, \ldots, t_{n}: A_{n}, t_{n+1}, \ldots t_{m}$ builds in parallel the terms $t_{1}, \ldots, t_{m}$, following the logical structure of the intended communication pattern. The commas in the conclusion of a sequent, indeed, mean exactly parallel composition. Since the comma, in the final analysis, is the $\pa$ connective, we obtain the following typing rule:
\[ \infer[\pa \I ]{ \Gamma \ergo s\p t : A\pa B , \Delta }{ \Gamma
\ergo s: A , t: B , \Delta } \]
As the $\pa$ is commutative and associative, so it is the comma. Our reduction rules, indeed, treat parallel composition as commutative and associative, since they bypass the order and association of processes. Technically, however, the term $s \p t: A\pa B$ is different from the term $t\p s: B\pa A$, since $A\pa B$ and $B\pa A$ are logically equivalent, but not the same type. It would be possible to add explicit commutation and association congruences on the outermost parallel operators $|$, provided Subject Reduction is restated modulo logical equivalence. Since there is no computational gain in adding such congruences, we do not consider them.

\subsection{Linear $\lambda$-calculus}
$\lamp$ contains the linear $\lambda$-calculus as subsystem.  On one hand, function application and variable declaration are respectively given by the rule
\[\infer[\lo \E]{ \Gamma , \Sigma  \ergo st:  B  , \Delta , \Theta }{     \Gamma \ergo s: A
    \lo B , \Delta  &&     \Sigma \ergo t: A , \Theta }\]
    and the axiom
      \[x: A\ergo x: A\]
      as usual. On the other hand,  the treatment of linear implication is the main novelty of the type system of $\lamp$. Namely, the linear implication  introduction rule is broken 
 in two rules. The first is intuitionistic and yields  functional abstraction, the second is classical and yields communication. Functional abstraction $\lambda$ is obtained by the rule:
\[\infer[\lo \I\ (\mbox{ if $x$ occurs in $t$})]{    \Gamma \ergo \lambda{x}\,  t: A \lo B , \Delta }{   
\Gamma , x: A \ergo t: B  , \Delta  }\]
Such a rule has been studied by~\cite{paivalinear} in the context of multiple-conclusion intuitionistic linear logic. Indeed, when $x$ occurs in $t$, the term $\lambda x\, t$ is a standard linear $\lambda$-term.
The application of a function to an argument generates in $\lamp$ the standard $\lambda$-calculus reduction
    \[(\lambda x\,u ) t\mapsto u[t/x]\]
  where $u[t/x]$ is the term obtained from $u$ by replacing the occurrence of $x$ with $t$. 

\subsection{Communication}

The other rule for the linear implication introduction is:
\[\infer[\lo \I\ (\mbox{ if $x$ occurs  in $\Delta$})] {    \Gamma \ergo \send{x}  t: A \lo B , \Delta }{   
\Gamma , x: A \ergo t: B  , \Delta  }
\]
In this case, the variable $x$ does not occur in $t$, therefore the term $\send{x} t$ is no more a $\lambda$-term and behaves like an output channel with continuation $t$. Namely, when $\send{x} t$ is
applied to an argument $u$, a communication is triggered: the
term $u$ is transmitted along the channel $x$ to another process inside $\Delta$ and will replace $x$ inside that
process. We thus have in $\lamp$ the reduction rules:
\[  \calc [\sendm{x}{s}{u}]  \p  \cald [x]   \quad
\red \quad \calc [s]  \p \cald [u]  \]
\[  \calc [x]  \p \cald [\sendm{x}{s}{u}]   \quad
\red \quad  \calc [u]  \p \cald [s]  \]
where  $\calc[\; ]$ and $\cald[\; ]$ are arbitrary contexts, as usual defined as follows.

\begin{definition}[Contexts]\label{def-context} A \textbf{context} $\calc[\; ]$ is a term
containing  a variable $[\;]$ occurring exactly once. For any term $t$ we denote by
$\calc[t]$ the result of replacing $[\;]$ by $t$ in $\calc[\; ]$. 
$\calc[\; ]$ is \textbf{simple} if it does not contain subterms of the form $\cald[\;]\p u$ or  $u \p\cald[\;]$ .
\end{definition} 
\noindent In the reduction above, by definition of context, only one occurrence of $x$ is replaced by $u$. Therefore, if the context $\cald[\;]$ is not simple, then the term $\cald[x]$ is a parallel composition of several processes and only one of them will actually receive  $u$. For instance, we have 
\[\sendm{x}{\termfal}{u}\p (\lambda y\, y \p z\, x) \red {\termfal}\p (\lambda y\, y \p z\, u).\] 
Unlike in pure $\pi$-calculus, the output operator $\sendm{x}{s}{u}$ may transmit its message even if it is surrounded by a potentially non-empty context. This is necessary for fully normalizing the proofs: the alternative is to give up on a full correspondence with the normalization/cut-elimination process, as done in Wadler's typed $\pi$-calculus CP \cite{Wadler2012}.
In our case, 
 it might happen that some variable $y$ of the message $u$ lies under the scope of an operator $\lambda y$ occurring in the context. There is no problem, though: 
 the operator $\lambda y$ is just a notation which automatically becomes $\overline{y}$, if the need arises -- namely, if $y$ changes location.  For instance, 
 let us consider the natural reduction strategy that avoids to reduce under $\lambda$:  \[(\lambda y\, \sendm{x}{\termfal}{(y\, w)})v \p  x\ \red\ \sendm{x}{\termfal}{(v\, w)} \p  x\ \red\ \termfal \p v\, w\] 
If we instead first fire the communication reduction along the channel $x$, we obtain the very same result, but by a different mechanism:
\[(\lambda y\, \sendm{x}{\termfal}{(y\, w)})v \p x\;
=\;(\send{y}\sendm{x}{\termfal}{(y\, w)})v \p x \;\red\; (\send{y}\, {\termfal})v
\p y\, w \;=\; \sendm{y}{\termfal}{v} \p y\, w\;\red\; \termfal \p v\, w\]
What happens in this second reduction is that the value of $y$, still
not known during the first communication, is automatically redirected
and communicated to the new location of $y$ by the second
communication.  Therefore, the code mobility issues created by
closures, solved in~\cite{ACG2017} in a more complicated way,
literally disappear in $\lamp$.

Similarly, the ``capture'' of variables during communications by output operators $\overline{x}$ creates no issue. For instance, the result of the reduction above communicating first along $x$, then along $y$:
$$\sendm{x}{y}{v} \p \sendm{y}{\termfal}{x}\;\red\; y \p \sendm{y}{\termfal}{v}\;\red\; v\p \termfal$$
can also be obtained by communicating first along $y$, then along $x$, thus letting the variable $x$ to be \emph{captured} by the $\overline{x}$ operator, which then automatically becomes $\lambda x$:
$$\sendm{x}{y}{v} \p \sendm{y}{\termfal}{x}\;\red\; \sendm{x}{x}{v} \p \termfal\; =\; (\lambda x\, x) v \p \termfal\;\red\; v\p \termfal$$

\subsection{Output channels without continuation}
The elimination rule for the connective $\pa$ types binary output channels with no associated continuation, a kind of  operator also found in the asynchronous $\pi$-calculus \cite{Bou92}:
\[ \infer[\pa \E ]{\Gamma , \Sigma_1 ,\Sigma _2 \ergo \distp{x}{y}  s:\fal , \Delta , \Theta _1 ,\Theta _2
}{ \Gamma \ergo s: A \pa B , \Delta && \Sigma_1 , x: A \ergo
\Theta _1 && \Sigma _2 , y:B \ergo \Theta _2 }
\]
When $s$ is a parallel composition $u \p v$, the term $\distp{x}{y}  s$ transmits $u$ and $v$ respectively along the channel $x$ and $y$ and then terminates with no value, as reflected by its type $\bot$. As a result, exactly one occurrence of $x$ will be replaced by $u$ and exactly one occurrence of $y$ will be replaced by $v$, in two different processes or in the same. The correspondent reduction rules of $\lamp$ are:
\[\calc [\distp{x}{y} (s \p t)] \p \cald[x][y] \quad \red \quad \calc
[\termfal]\p \cald[s][t]
 \]
\[ \cald[x][y]\p \calc [\distp{x}{y} (s \p t)] \quad \red \quad
\cald[s][t]\p \calc [\termfal]
 \]
  \[ \cald[x] \p \calc [\distp{x}{y} (s \p t)]\p \cale[y] \quad \red
\quad \cald[s]\p \calc [\termfal]\p \cale[t]
 \]
 \[ \cald[y] \p \calc [\distp{x}{y} (s \p t)]\p \cale[x] \quad \red
\quad \cald[t]\p \calc [\termfal]\p \cale[s]
 \] The first two reduction rules concern the case in which the
variables $x$ and $y$ occur in the same context. The last two address
the case in which the variables $x$ and $y$ occur in two different
contexts.

\subsection{Process Termination}
Classical logic is obtained by a combination of the classical linear implication introduction rule and the $\bot$-introduction rule:
\[\infer[\fal\I]{\Gamma \ergo \Delta , \termfal :\fal }{\Gamma \ergo
\Delta}\]
which introduces the terminated process $\termfal$, which does nothing.
Indeed, we can derive the linear excluded middle as follows:
\begin{prooftree}
\AxiomC{$x: A \ergo x: A$}
\UnaryInfC{$x:A \ergo x: A,\, \termfal: \fal$}
\UnaryInfC{$\ergo x: A,\, \send{x}\termfal: A\lo\fal$}
\UnaryInfC{$\ergo x \p \send{x}\termfal: A \pa (A\lo\fal)$}
\end{prooftree}
Since terms of type $\bot$ return no value, we have the typing rule
\[ \infer[\fal \E ]{\Gamma \ergo \abort(t), \Delta}{\Gamma \ergo  t: \fal , \Delta}
\]
which introduces the construct $\abort(t)$, whose intended meaning is to execute $t$ and then terminate the computation with no return value, hence it has no type.

\subsection{The reduction rules of $\lamp$}\label{sec:reductions}

The complete list of reduction rules of $\lamp$ is presented in
Table~\ref{tab:reductions} for the reader's convenience.  
\begin{table}[h!]  \centering
  \begin{footnotesize}
\hrule\vspace{-8pt}
\[(\lambda x\,u ) t\quad \red\quad  u[t/x]\]

\vspace{-8pt}
\[  \calc [\sendm{x}{s}{u}]  \p  \cald [x]   \quad
\red \quad \calc [s]  \p \cald [u]  \]
\[  \calc [x]  \p \cald [\sendm{x}{s}{u}]   \quad
\red\quad   \calc [u]  \p \cald [s]  \]

\vspace{-8pt}
\[\calc [\distp{x}{y} (s \p t)] \p \cald[x][y]  
 \quad \red \quad  \calc [\termfal]\p  \cald[s][t]
 \]
\[ \cald[x][y]\p  \calc [\distp{x}{y} (s \p t)]
 \quad \red \quad   \cald[s][t]\p \calc [\termfal]
 \] 
  \[ \cald[x]  \p \calc [\distp{x}{y} (s \p t)]\p  \cale[y]
 \quad \red \quad   \cald[s]\p  \calc [\termfal]\p \cale[t]  
 \]
 \[ \cald[y]  \p \calc [\distp{x}{y} (s \p t)]\p  \cale[x]
 \quad \red \quad   \cald[t]\p  \calc [\termfal]\p \cale[s]  
 \]
\hrule
\end{footnotesize}
\caption{Reduction rules for $\NLL$ terms.}
\label{tab:reductions}
\end{table}

\noindent As usual, for any context $\calc[\;]$, we adopt
the reduction scheme $\calc [t] \red \calc[u]$ whenever $t\red u
$. We denote by  $\red ^*$ the reflexive and transitive
closure of the one-step reduction $\red$. As usual, we shall employ the $\lambda$-calculus concepts of normal form and strong normalization.

\begin{definition}[Normal Forms and Strongly Normalizable Terms]\label{def:formnorm}\mbox{}
 \begin{itemize}
\item  A \textbf{redex} is a term $u$ such that $u\mapsto v$ for some
  $v$ and reduction in Table~\ref{tab:reductions}. 
A term $t$ is called a \textbf{normal form} or, simply, \textbf{normal}, 
if there is no $u$ such that $t\mapsto u$.
\item  
A finite or infinite sequence of terms
$u_1,u_2,\ldots,u_n,\ldots$ is said to be a \textbf{reduction} of $t$, if
$t=u_1$, and for all  $i$, $u_i \mapsto
u_{i+1}$.
 A  term $u$ of $\lamp$ is  \textbf{normalizable} if there is a finite reduction of $u$ whose last term is normal and is \textbf{strong normalizable} if every reduction of $u$ is finite. \end{itemize}\end{definition}

\section{Programming with $\lamp$}\label{sec:examples}

In order to acquire familiarity with the main features of $\lam$, we discuss now some programming examples. Namely, we shall see how to implement client-server communication, synchronization, cyclic communication patterns and channel sharing.

\begin{example}[Client-Server: Request/Answer]
In this example, we discuss how to represent in $\lamp$ a simple communication protocol, consisting in a request/answer message exchange. A server hosts an online catalogue, mapping product names to prices. A client transmits a string $\mathsf{prod}: \mathsf{String}$ to the server, representing a product name whose price the client wishes to know. The server answers the client request by sending back the price of the product $\mathsf{prod}$, computed by the function $\mathsf{cost}: \mathsf{String}\lo \mathbb{N}$.  A naive way to implement client and server would be
$$\overbrace{\sendm{x}{y}{\mathsf{prod}}}^{CLIENT}  \p  \overbrace{\sendm{y}{\termfal}{(\mathsf{cost}\, x)}}^{SERVER}$$
Although, as  expected
$$\sendm{x}{y}{\mathsf{prod}} \p  \sendm{y}{\termfal}{(\mathsf{cost}\, x)}
\ \mapsto {y} \p  \sendm{y}{\termfal}{(\mathsf{cost}\, \mathsf{prod})}\mapsto \  {y} \p  \sendm{y}{\termfal}{(\mathsf{price})}  \mapsto\ {\mathsf{price}} \p  {\termfal} $$
 we observe that, this way, synchronization is implemented poorly.  As we see below, the server may send to the client the message $(\mathsf{cost}\, x)$ before it receives any request whatsoever!
$$\sendm{x}{y}{\mathsf{prod}} \p  \sendm{y}{\termfal}{(\mathsf{cost}\, x)}
\ \mapsto  
\sendm{x}{\mathsf{cost}\, x}{\mathsf{prod}} \p {\termfal}
= \ (\lambda x\, \mathsf{cost}\, x){\mathsf{prod}} \p {\termfal}
\mapsto^{*}{\mathsf{price}} \p  {\termfal} $$
What we expect, instead, is that only the request of the client can trigger computation and transmission of the answer by the server. In order to force that, the trick  is to encapsulate the message $\mathsf{prod}$ into the $\lambda$-term $\lambda a \lambda b\, a\, (b\, \mathsf{prod})$, whose task is to take as input a server channel $a$, a continuation $b$ and apply the channel $a$ to $b\,\mathsf{prod}$. We thus implement client and server as follows:
$$\overbrace{\sendm{x}{y}{(\lambda a \lambda b\, a\, (b\, \mathsf{prod}))}}^{CLIENT} \p \overbrace{x\, (\send{y}{\termfal})\, \mathsf{cost}}^{SERVER}$$
Since  now the server output channel $\send{y}\termfal$ is not applied to any argument, the server has no choice but to wait for the client request:
\begin{align*}
&\sendm{x}{y}{(\lambda a \lambda b\, a\, (b\, \mathsf{prod}))} \p x\, (\send{y}{\termfal})\, \mathsf{cost}\\
&\mapsto {y} \p (\lambda a \lambda b\, a\, (b\, \mathsf{prod}))\, (\send{y}{\termfal})\, \mathsf{cost}\\
&\mapsto^{*} {y} \p  (\sendm{y}{\termfal}{(\mathsf{cost}\, \mathsf{prod})}\\
&\mapsto^{*} {y} \p  \sendm{y}{\termfal}{\mathsf{price}}\\
&\mapsto \mathsf{price} \p  {\termfal}\\
\end{align*}
\end{example}
By defining $X:=(\mathbb{N}\lo \fal)\lo  (\mathsf{String}\lo \mathbb{N})\lo \fal $, we can type the process above as follows:
\begin{footnotesize}
\begin{prooftree}
\AxiomC{$x:  X\ergo x:X $}
\AxiomC{$ \ergo y: \mathbb{N},\, \send{y} \termfal: \mathbb{N}\lo \fal$}
\BinaryInfC{$x:  X\ergo y: \mathbb{N},\, x\,(\send{y} \termfal): (\mathsf{String}\lo \mathbb{N})\lo \fal$}
\AxiomC{$\ergo \mathsf{cost}: \mathsf{String}\lo \mathbb{N}$}
\BinaryInfC{$x:X \ergo y: \mathbb{N},\, x\,(\send{y} \termfal)\,\mathsf{cost}:  \fal$}
\UnaryInfC{$\ergo \send{x}y: X\lo\mathbb{N},\, x\,(\send{y} \termfal)\, \mathsf{cost}:  \fal$}
\AxiomC{$  \ergo (\lambda a \lambda b\, a\, (b\, \mathsf{prod})): X $}
\def\defaultHypSeparation{\hskip .01in}
\BinaryInfC{$\ergo \sendm{x}{y}{(\lambda a \lambda b\, a\, (b\, \mathsf{prod}))}: \mathbb{N},\, x\,(\send{y} \termfal)\, \mathsf{cost}:  \fal$}
\UnaryInfC{$\ergo \sendm{x}{y}{(\lambda a \lambda b\, a\, (b\, \mathsf{prod}))}\p x\,(\send{y} \termfal)\, \mathsf{cost}:  \mathbb{N}\pa \fal$}
\end{prooftree}
\end{footnotesize}

\begin{example}[Client-Server: Dialogue]

We present now a continuation of the previous example which models a
longer interaction.  We want to represent the following transaction of
an online sale. As before, a buyer transmits to a seller a product
name $\mathsf{prod}:\mathsf{String}$, the seller computes by the
function $\mathsf{cost}$ the monetary cost $\mathsf{price}:\mathbb{N}$
of $\mathsf{prod}$ and communicates it to the buyer.  Then the buyer
apply a function $\mathsf{pay}: \mathbb{N}\rightarrow \mathsf{String}$
to $\mathsf{price}$, that will transmit to the server the credit card
number $\mathsf{card}: \mathsf{String}$ if the client wishes to buy
the product, the empty string otherwise. By comparison with the
previous example, we must add a new communication channel
$\overline{z}$ for the third communication, because in $\lamp$ each
communication requires a fresh channel. As before, in order to
implement synchronization, messages $M$ will be always transmitted as
$(\lambda a \lambda b\, a\, (b\, M))$:

$$\overbrace{\sendm{x}{y\, (\send{z}\termfal)\, \mathsf{pay} }{(\lambda a \lambda b\, a\, (b\, \mathsf{prod}))}}^{CLIENT} \p \overbrace{x\, (\send{y}{z})\, (\lambda c \lambda a' \lambda b'\, a'\, (b'\, \mathsf{cost}\, c))}^{SERVER}$$
We can observe the following  interaction:
\begin{align*}
&\sendm{x}{y\, (\send{z}\termfal)\, \mathsf{pay} }{(\lambda a \lambda b\, a\, (b\, \mathsf{prod}))} \p \, x\, (\send{y}{z})\, (\lambda c \lambda a' \lambda b'\, a'\, (b'\, \mathsf{cost}\, c))\\
&\mapsto {y\, (\send{z}\termfal)\, \mathsf{pay} } \p (\lambda a \lambda b\, a\, (b\, \mathsf{prod}))\, (\send{y}{z})\, (\lambda c \lambda a' \lambda b'\, a'\, (b'\, \mathsf{cost}\, c))\\
&\mapsto^{*} {y\, (\send{z}\termfal)\, \mathsf{pay} } \p  \sendm{y}{z}{((\lambda c \lambda a' \lambda b'\, a'\, (b'\, \mathsf{cost}\, c))\mathsf{prod})}\\
&\mapsto^{*} {y\, (\send{z}\termfal)\, \mathsf{pay} } \p  \sendm{y}{z}{( \lambda a' \lambda b'\, a'\, (b'\, \mathsf{price}))}\\
&\mapsto {( \lambda a' \lambda b'\, a'\, (b'\, \mathsf{price}))\, (\send{z}\termfal)\, \mathsf{pay} } \p  {z}\\
&\mapsto^{*} \sendm{z}{\termfal}{(\mathsf{pay}\, \mathsf{price} )} \p  {z}\\
&\mapsto \sendm{z}{\termfal}{\mathsf{card}} \p  {z}\\
&\mapsto {\termfal} \p  {\mathsf{card}}\\
\end{align*}
By defining\\ $X:=(Y\lo \mathsf{String})\lo (\mathsf{String}\lo Y)\lo \mathsf{String} $ \\ $Y:=( \mathsf{String}\lo \fal)\lo  (\mathbb{N}\lo \mathsf{String})\lo \fal $\\ we can type the process above as follows:

\begin{footnotesize}
\begin{prooftree}
\AxiomC{$x:  X\ergo x:X $}
\AxiomC{ $\ergo z: \mathsf{String},\, \send{z} \termfal: \mathsf{String}\lo \fal$}
\AxiomC{$y: Y\ergo y: Y$}
\BinaryInfC{$y: Y\ergo  y\, (\send{z} \termfal): (\mathbb{N}\lo \mathsf{String})\lo \fal, \,z: \mathsf{String}$}
\UnaryInfC{$\ergo  y\, (\send{z} \termfal): (\mathbb{N}\lo \mathsf{String})\lo \fal, \,\send{y}z: Y\lo\mathsf{String}$}
\AxiomC{$\ergo \mathsf{pay}: \mathbb{N}\lo \mathsf{String}$}
\def\defaultHypSeparation{\hskip .01in}
\BinaryInfC{$\ergo  y\, (\send{z} \termfal)\, \mathsf{pay}: \fal, \,\send{y}z: Y\lo\mathsf{String}$}
\BinaryInfC{$x:  X\ergo   y\, (\send{z} \termfal)\, \mathsf{pay}: \fal, \,x\, (\send{y}z):(\mathsf{String}\lo Y)\lo \mathsf{String}\qquad \ergo \lambda c \lambda a' \lambda b'\, a'\, (b'\, \mathsf{cost}\, c): \mathsf{String}\lo Y$}
\UnaryInfC{$x:X \ergo  y\, (\send{z} \termfal)\, \mathsf{pay}: \fal, \,x\, (\send{y}z)(\lambda c \lambda a' \lambda b'\, a'\, (b'\, \mathsf{cost}\, c)):\mathsf{String}$}
\UnaryInfC{$\ergo  \send{x}{y\, (\send{z} \termfal)\, \mathsf{pay}}: X\lo\fal, \,x\, (\send{y}z)(\lambda c \lambda a' \lambda b'\, a'\, (b'\, \mathsf{cost}\, c)):\mathsf{String}\qquad \ergo  \lambda a \lambda b\, a\, (b\, \mathsf{prod}): X$}
\UnaryInfC{$\ergo  \sendm{x}{y\, (\send{z} \termfal)\, \mathsf{pay}}{(\lambda a \lambda b\, a\, (b\, \mathsf{prod}))}: \fal,\, \,x\, (\send{y}z)(\lambda c \lambda a' \lambda b'\, a'\, (b'\, \mathsf{cost}\, c)):\mathsf{String}$}
\UnaryInfC{$\ergo  \sendm{x}{y\, (\send{z} \termfal)\, \mathsf{pay}}{(\lambda a \lambda b\, a\, (b\, \mathsf{prod}))} \p \,x\, (\send{y}z)(\lambda c \lambda a' \lambda b'\, a'\, (b'\, \mathsf{cost}\, c)):\fal \pa \mathsf{String}$}
\end{prooftree}
\end{footnotesize}
\end{example}

\begin{example}[Cyclic Communication]
Unlike in the session typed $\pi$-calculi in \cite{CairesPfenning} and \cite{Wadler2012},  in $\lamp$ one can  type in a natural way  cyclic communication patterns, as the following example shows. A client possesses some secret information $M: \mathsf{String}$, which it wishes to encrypt. For added security, the client desires two encryptions of the message. A couple of servers, by joining forces, offer this kind of double-encryption service. Therefore, the client sends $M$ to the first server, which encrypts $M$ by applying the function $\mathsf{enc}_{1}: \mathsf{String}\lo \mathsf{String} $ and then transmits the result to the second server, which in turn encrypts it by applying the function $\mathsf{enc}_{2}: \mathsf{String}\lo \mathsf{String} $ and finally transmits the result to the client. We implement client and servers as follows:

$$\overbrace{\sendm{x}{z}{M}}^{CLIENT} \p \overbrace{\sendm{y}{\termfal}{(\mathsf{enc}_{1}\, x)}}^{SERVER\ 1}\p \overbrace{\sendm{z}{\termfal}{(\mathsf{enc}_{2}\, y)}}^{SERVER\ 2}$$
We observe the following interaction:
\begin{align*}&\sendm{x}{z}{M} \p \sendm{y}{\termfal}{(\mathsf{enc}_{1}\, x)}\p \sendm{z}{\termfal}{(\mathsf{enc}_{2}\, y)}\\
&\mapsto z \p \sendm{y}{\termfal}{(\mathsf{enc}_{1}\, M)}\p \sendm{z}{\termfal}{(\mathsf{enc}_{2}\, y)}  \\ 
&\mapsto z \p \sendm{y}{\termfal}{M'}\p \sendm{z}{\termfal}{(\mathsf{enc}_{2}\, y)}  \\
&\mapsto z \p {\termfal}\p \sendm{z}{\termfal}{(\mathsf{enc}_{2}\, M')}  \\ 
&\mapsto z \p {\termfal}\p \sendm{z}{\termfal}{M''}  \\
&\mapsto M'' \p {\termfal}\p \termfal  
\end{align*}
We can type the process above as follows:
\begin{footnotesize}
\begin{prooftree}
\AxiomC{$ \ergo z: \mathsf{String},\, \send{z} \termfal: \mathsf{String}\lo \fal$}
\AxiomC{$y:  \mathsf{String}\ergo \termfal: \fal,\, \mathsf{enc_{2}}\,y: \mathsf{String} $}
\BinaryInfC{$y:  \mathsf{String}\ergo z: \mathsf{String},\, \termfal: \fal,\, \sendm{z}{ \termfal}{(\mathsf{enc_{2}}\,y)}: \fal$}
\UnaryInfC{$\ergo z: \mathsf{String},\,\send{y} \termfal: \mathsf{String}\lo\fal,\, \sendm{z}{ \termfal}{(\mathsf{enc_{2}}\,y)}: \fal $}
\AxiomC{$x: \mathsf{String} \ergo  \mathsf{enc_{1}}\,x: \mathsf{String} $}
\def\defaultHypSeparation{\hskip .01in}
\BinaryInfC{$x: \mathsf{String}\ergo z: \mathsf{String},\,\sendm{y}{\termfal}{(\mathsf{enc_{1}}\,x)}: \fal,\, \sendm{z}{ \termfal}{(\mathsf{enc_{2}}\,y)}: \fal $}
\UnaryInfC{$\ergo \send{x}z: \mathsf{String}\lo\mathsf{String},\,\sendm{y}{\termfal}{(\mathsf{enc_{1}}\,x)}: \fal,\, \sendm{z}{ \termfal}{(\mathsf{enc_{2}}\,y)}: \fal \qquad  \ergo M: 	\mathsf{String} $}
\UnaryInfC{$\ergo \sendm{x}{z}{M}: \mathsf{String},\,\sendm{y}{\termfal}{(\mathsf{enc_{1}}\,x)}: \fal,\, \sendm{z}{ \termfal}{(\mathsf{enc_{2}}\,y)}: \fal $}
\UnaryInfC{$\ergo \sendm{x}{z}{M} \p \sendm{y}{\termfal}{(\mathsf{enc_{1}}\,x)}:\mathsf{String}\pa \fal,\, \sendm{z}{ \termfal}{(\mathsf{enc_{2}}\,y)}: \fal $}
\UnaryInfC{$\ergo \sendm{x}{z}{M} \p \sendm{y}{\termfal}{(\mathsf{enc_{1}}\,x)}\p \sendm{z}{ \termfal}{(\mathsf{enc_{2}}\,y)}: (\mathsf{String}\pa \fal)\pa\fal $}
\end{prooftree}
\end{footnotesize}

%

\end{example}

\begin{example}[Channel Transmission]
As $\pi$-calculus, $\lamp$  supports communication of channel names and thus dynamic communication patterns, as we see in the following example. A server offers a printing service, which is hosted on another machine, connected to the server. In order to exploit the service an access code is required. A consumer, which wants to print the string $M$, sends its access code $\mathsf{access}: \mathsf{String}$ to the server, which checks it by the function $\mathsf{check}$. Upon success, the server transmits to the consumer the channel $\overline{z}$ along which the printer offers its services, so that finally the consumer can send $M$ to the printer. We model printer, server and client as follows:
$$\overbrace{z\, \mathsf{print}}^{PRINTER} \p \overbrace{(\mathsf{check}\, x)(\send{y}\termfal)(\send{z}\termfal)}^{SERVER}\p \overbrace{\sendm{x}{y\, (\lambda a\, a\, M)}{\mathsf{access}}}^{CONSUMER}$$
We observe the following interaction:
\begin{align*}
&{z\, \mathsf{print}} \p {(\mathsf{check}\, x)(\send{y}\termfal)(\send{z}\termfal)}\p {\sendm{x}{y\, (\lambda a\, a\, M)}{\mathsf{access}}}\\
&\mapsto z\, \mathsf{print} \p {(\mathsf{check}\, \mathsf{access})(\send{y}\termfal)(\send{z}\termfal)}\p y\, (\lambda a\, a\, M)\\
&\mapsto  z\, \mathsf{print} \p \sendm{y}{\termfal}{(\send{z}\termfal)}\p y\, (\lambda a\, a\, M)\\
&\mapsto  z\, \mathsf{print} \p {\termfal}\p \sendm{z}{\termfal}{(\lambda a\, a\, M)}\\
&\mapsto  (\lambda a\, a\, M)\, \mathsf{print} \p {\termfal}\p {\termfal}\\
&\mapsto   \mathsf{print}\, M \p {\termfal}\p {\termfal}
\end{align*}
By defining $C:=((\mathsf{String}\lo \fal)\lo \fal)\lo \fal$, we can type the above process as follows:

\begin{footnotesize}
\begin{prooftree}
\AxiomC{$ x: \mathsf{String}\ergo \mathsf{check}\, x: C\lo C\lo \fal$}
\AxiomC{$y: C\ergo \termfal: \fal, \, y\, (\lambda a\, a\, M): \fal $}
\UnaryInfC{$\ergo \send{y}\termfal: C, \, y\, (\lambda a\, a\, M): \fal $}
\BinaryInfC{$x: \mathsf{String}\ergo  (\mathsf{check}\, x)(\send{y}\termfal): C\lo \fal, \, y\, (\lambda a\, a\, M): \fal $}
\UnaryInfC{ $\ergo (\mathsf{check}\, x)(\send{y}\termfal): C\lo \fal, \, \send{x}y\, (\lambda a\, a\, M): \mathsf{String}\lo\fal $}
\AxiomC{$ \ergo  \mathsf{access}:\mathsf{String} $}
\def\defaultHypSeparation{\hskip .01in}
\BinaryInfC{$\ergo (\mathsf{check}\, x)(\send{y}\termfal): C\lo \fal, \, \sendm{x}{y\, (\lambda a\, a\, M)}{\mathsf{access}}: \fal \qquad  \ergo z\, \mathsf{print}:\fal, \send{z}\termfal: C $}
\UnaryInfC{$\ergo z\, \mathsf{print}:\fal,\, (\mathsf{check}\, x)(\send{y}\termfal)(\send{z}\termfal): \fal, \, \sendm{x}{y\, (\lambda a\, a\, M)}{\mathsf{access}}: \fal $}
\UnaryInfC{$\ergo z\, \mathsf{print} \p (\mathsf{check}\, x)(\send{y}\termfal)(\send{z}\termfal): \fal \pa\fal, \, \sendm{x}{y\, (\lambda a\, a\, M)}{\mathsf{access}}: \fal $}
\UnaryInfC{$\ergo z\, \mathsf{print} \p (\mathsf{check}\, x)(\send{y}\termfal)(\send{z}\termfal)\p \sendm{x}{y\, (\lambda a\, a\, M)}{\mathsf{access}}: (\fal\pa\fal)\pa\fal $}
\end{prooftree}
\end{footnotesize}
\end{example}

\section{Intermezzo: Synchronous Communication in $\lamp$}
\label{sec:synchrony}

In $\pi$-calculus the actions of both sending and receiving messages are \emph{synchronous} and \emph{blocking}. They are synchronous, because they 
require that the  sender and the receiver synchronize;  they are \emph{blocking}, because the execution of both the sender and the receiver is blocked until the message is actually transmitted. 
Therefore, on one hand, if there is no receiver listening to the channel, a process can neither transmit its message along the channel nor proceed with its execution. On the other hand, 
a process which does listen cannot proceed with its execution until it receives the message it is waiting for. In $\pi$-calculus, synchronicity is implemented by a construct for sending $\overline{x}m.\,\_$ and a construct for receiving  $x(y).\,\_$
which need to be outermost in the connected processes in order to activate the reduction:
$$\overline{x}m.P \ | \ x(y).Q\ \mapsto\ P \ | \ Q[m/y]$$
The blocking nature of the actions is implemented by forbidding to reduce inside $P$ and $Q$ before the actions are executed, that is, by forcing a reduction strategy. 

Non-blocking actions have been advocated by various authors (e.g. \cite{Merro}, \cite{KMP2019}) and are necessary for a full correspondence between cut-elimination and process execution. In $\lamp$, the fact that sending and receiving are non-blocking is implemented by just removing the input construct $x(y)$, so in $\lamp$ we have :
$$\overline{x}m.P \ |\ Q \mapsto P \ | \ Q[m/x]$$
provided that $x$ occurs in $Q$.

\subsection{The call-by-value $\lamp$}

Although $\lamp$ is naturally asynchronous, we can   model synchronous and blocking actions also in $\lamp$, by defining reduction to automatically follow a call-by-value discipline.

\begin{definition}[Value] A \emph{value} is any term of one of the
following forms: \[\send{x} t 
\qquad\qquad\qquad s\p t \]\end{definition}

\begin{definition}[Call-by-value Contexts]
We define a \textbf{call-by-value evaluation context} by the following grammar: 
$$E\;\;:=\;\; [\;] \;\;\p\;\; E\, u \;\;\p\;\; V\, E \;\;\p\;\; \distp{x}{y}{E} \;\;\p\;\; \abort(E) $$
where $V$ is any value. 
\end{definition}

In order to obtain a call-by-value version of $\lamp$, we rewrite the main reduction rules of $\lamp$ as following:

\[\ldots\p E[(\lambda x\,u ) V]\p \ldots \mapsto^{cbv} \ldots \p E[u[V/x]]\p \ldots\]
\[ \dots\p E [\sendm{x}{s}{V}]  \p \ldots \p E'[x]\p \dots   \quad
\red^{cbv} \quad\dots\p E[s]  \p \ldots \p E'[V]\p\dots  \]
\[ \dots\p E'[x]  \p \ldots \p E [\sendm{x}{s}{V}]\p\dots   \quad
\red^{cbv} \dots\p E'[V]  \p \ldots \p E[s]\p\dots  \]
where $V$ is a value and $E, E'$ call-by-value evaluation contexts. The idea is that call-by-value contexts bring back a notion  of order among operations, which can be used to represent the sequential facet of computation. Thus a term $E'[x]$ represents a process whose next operation is a read operation on the channel $x$, followed by the operations contained in $E'[\;]$. Similarly to Wadler's principal-cut-elimination strategy for CP, the $\mapsto^{cbv}$ reduction acts only on the top level parallel operators.

We can now define in $\lamp$ the synchronous input-construct as follows:
$$x(y).u := (\lambda y\, u)\, x$$
Indeed, the term $x(y).u$ is always stuck: since $u$ is located under a $\lambda$, no reduction can be performed inside $u$; since $x$ is not a value,  the redex $(\lambda y\, u)\, x$ cannot be contracted either.
For similar reasons, the term $\sendm{x}{u}{V}= (\send{x}u)V$ is  stuck as well: since $V$ is a value, no reduction can be performed inside $V$; since $u$ is located under an output operator, no reduction can be performed inside $u$ either. Thus, the only reduction that could be fired is the transmission of $V$ along the channel $x$, which however requires a suitable receiver, namely a process whose call-by-value evaluation cannot further proceed: the process requires an input value and waits for it. Indeed, in the following  configuration, we can reduce
$$\sendm{x}{u}{V} \p  x(y).w\ =\ \sendm{x}{u}{V}\p (\lambda y\, w) x\ \mapsto^{cbv} u\p (\lambda y\, w) V\ \mapsto^{cbv}\ u\p w[V/x]$$ 
In the configuration
$$\sendm{x}{u}{V} \p (\lambda z\, z\, x(y).w ) (\lambda k\, k)$$
however,
the process $\sendm{x}{u}{V}$ cannot transmit its message $V$ until $(\lambda z\, z\, x(y).w ) (\lambda k\, k)$ is ready to receive, which will happen as soon as its head redex will be contracted and the value of $x(y).w$ will be really needed. Thus, after one reduction step, we have
$$\sendm{x}{u}{V} \p (\lambda k\, k)\, x(y).w\ \mapsto^{cbv} u \p (\lambda k\, k)((\lambda y\, w)\,V)\ \mapsto^{cbv}\   u \p (\lambda k\, k)\, w[V/x] \ \mapsto^{cbv}\   u \p w[V/x]$$

\section{Soundness and Completeness} \label{sec:soundness-completeness}

We show now that the typing system $\NLL$ of $\lamp$  captures exactly classical multiplicative linear logic. Namely, we show that $\NLL$ is equivalent to the sequent calculus $\SLL$
for classical multiplicative linear logic. The cut
rule and the logical rules of $\SLL $ -- namely, the rules introducing
the logical connectives of linear logic on the left or on the right --
are the standard double-sided sequent rules for classical linear
logic. The initial sequent $\fal \ergo $ corresponds to the initial
sequent $\ergo 1$. We use the former since we do not have an explicit
duality operator. The only non-standard rule of $\NLL$ is $\fal
\E$. It is clearly sound since it captures the neutrality of $\fal$
with respect to $\pa$, and hence to the comma on the right-hand side  of sequents. As we will see in the completeness proof, the rule $\fal \E$
is key in handling the occurrences of  $\fal$ introduced by the rules
$\pa\E$. 

\begin{table}[h!]
  \centering
  \begin{center}
  \hrule


\[\vcenter{\infer[\fal]{\Gamma \ergo \Delta ,
\fal }{\Gamma \ergo \Delta}} 
\qquad\qquad\qquad 
  \fal \ergo \]

  \[ 
\infer[\lo l]{ \Gamma , \Sigma , A \lo B \ergo \Delta , \Theta }{\Gamma  \ergo A,  \Delta &&  \Sigma ,  B \ergo \Theta} \qquad \qquad\quad \infer[\lo r]{ \Gamma \ergo A \lo B , \Delta }{\Gamma , A \ergo B  , \Delta  }
\]


\[
\infer[\pa l ]{  \Gamma ,  \Sigma , A\pa B \ergo 
      \Delta , \Theta}{\Gamma , A \ergo  \Delta  && \Sigma , B \ergo  \Theta} 
\qquad\qquad\qquad
  \infer[\pa r ]{ \Gamma \ergo A\pa B ,  \Delta}{\Gamma \ergo A ,  B ,  \Delta}
\]

\[\infer[\cut]{\Gamma , \Sigma \ergo \Delta , \Theta}{\Gamma \ergo A,  \Delta &&
\Sigma , A \ergo \Theta}\]

\hrule
\end{center}

\caption{The sequent calculus $\SLL$. }
\label{tab:sequent-calculus}
\end{table}

\begin{proposition}[Soundness] \label{proof:sound} For any sequent $S=
\Gamma \ergo \Delta $, if $S$ is derivable in $\NLL$, then $S$ is
derivable in $\SLL$.
\end{proposition}
\begin{proof} We proceed by induction on the number of rule
applications in the $\NLL$ derivation of $S$. Since the term assignment of the logical inferences is irrelevant for the argument, we ignore it in the present proof. 

If no rule is applied in the derivation, then $S= A \ergo A$, which is
an initial sequents of $\SLL$.

Suppose now that the statement holds for any $\NLL$ derivation
containing $n$ or less rule applications, we show the result for a
generic $\NLL$ derivation containing $n+1$ rule applications. We reason
by cases on the last rule applied in this $\NLL$ derivation.
\begin{itemize}

\item $\vcenter{\infer[\fal\I]{\Gamma \ergo \Delta , \fal }{\Gamma
\ergo \Delta}}$. Since the derivation of the premise of the rule
contain $n$ or less rule applications, by inductive hypothesis we know
that $\Gamma \ergo \Delta $ is derivable in $\SLL$. Then by the same
application of the $\fal$ rule we can derive $\Gamma \ergo \Delta ,
\fal$ also in $\SLL$.

\item $\vcenter{\infer[\fal\E]{\Gamma \ergo \Delta }{\Gamma
\ergo \Delta, \fal}}$. Since the derivation of the premise of the rule
contain $n$ or less rule applications, by inductive hypothesis we know
that $\Gamma \ergo \Delta, \fal $ is derivable in $\SLL$. Then by 
 $$ \vcenter{\infer[\cut]{\Gamma  \ergo \Delta }{\Gamma \ergo \Delta, \fal && \fal \ergo }}$$
we can derive $\Gamma \ergo \Delta$ also in $\SLL$.

\item $ \vcenter{ \infer[\lo \E]{ \Gamma , \Sigma \ergo B , \Delta ,
\Theta }{\Gamma \ergo A\lo B, \Delta && \Sigma \ergo A, \Theta}
}$. Since the derivations of the premises of the rule contain $n$ or
less rule applications, by inductive hypothesis we know that $\Gamma
\ergo A\lo B, \Delta$ and $ \Sigma \ergo A, \Theta$ are derivable in
$\SLL$. By
\[ \infer[\cut]{\Gamma , \Sigma \ergo B , \Delta ,
\Theta}{\Sigma \ergo A, \Theta && \infer[\cut]{\Gamma , A 
\ergo B, \Delta }{\Gamma
\ergo A\lo B, \Delta &&  \infer[\lo l]{A , A\lo B \ergo B}{ A  \ergo A &&  B\ergo B } }} \] we can derive $\Gamma , \Sigma \ergo B , \Delta ,
\Theta$ in $\SLL$.

\item $ \vcenter{ \infer[\lo \I]{ \Gamma \ergo A \lo B , \Delta
}{\Gamma , A \ergo B , \Delta }}$. Since the derivation of the
premise of the rule contain $n$ or less rule applications, by
inductive hypothesis we know that $\Gamma , A \ergo B , \Delta$ is
derivable in $\SLL$. By the rule application\[\infer[\lo r]{ \Gamma \ergo A \lo B , \Delta
}{\Gamma , A \ergo B , \Delta }\]we can derive $ \Gamma \ergo A \lo B , \Delta$ in $\SLL$.

%
%
%

\item $ \vcenter{ \infer[\pa \E]{ \Gamma ,\Sigma_1  ,\Sigma_2  \ergo \fal , \Delta , \Theta_1 , \Theta_2 }{\Gamma \ergo A\pa B , \Delta  
&& \Sigma _1   , A \ergo \Theta _1 && \Sigma _2   , B \ergo \Theta _2
}}$.  Since the derivations of the premises of the rule contain
$n$ or less rule applications, by inductive hypothesis we know that
$\Gamma \ergo A\pa B , \Delta  
$ and $ \Sigma _1   , A \ergo \Theta _1 $ and $  \Sigma _2   , B \ergo \Theta _2 $ are
derivable in $\SLL$. By the rule application
\[\infer[\fal]{ \Gamma ,\Sigma_1  ,\Sigma_2  \ergo \fal , \Delta , \Theta_1 , \Theta_2}{\infer[\cut]{\Gamma ,\Sigma_1  ,\Sigma_2  \ergo \Delta , \Theta_1 , \Theta_2}{\Gamma \ergo  A\pa B , \Delta && \infer[\pa l ]{ \Sigma_1  ,\Sigma_2 , A \pa B  \ergo \Theta_1 , \Theta_2 }{ \Sigma_1  , A \ergo \Theta _2 && \Sigma _2 , B \ergo \Theta _2} 
 }}\]we can
derive $ \Gamma ,\Sigma_1  ,\Sigma_2  \ergo \bot, \Delta , \Theta_1 , \Theta_2$ in $\SLL$.

\item $ \vcenter{ \infer[\pa \I ]{ \Gamma \ergo A\pa B , \Delta}{\Gamma
\ergo A , B , \Delta}}$. Since the derivation of the premise of the
rule contain $n$ or less rule applications, by inductive hypothesis we
know that $\Gamma , A \ergo A, B, \Delta $ is derivable in $\SLL$.  By the rule application\[\infer[\pa r ] { \Gamma \ergo A\pa B ,
\Delta}{\Gamma \ergo A , B , \Delta}\]we can derive $ \Gamma \ergo
A\pa B , \Delta$ in $\SLL$.
\end{itemize}
\end{proof}


The natural deduction calculus $\NLL$ corresponds very neatly to the
sequent calculus $\SLL$. Hence, a completeness proof of $\NLL$ with
respect to $\SLL$ is quite straightforward. The proof deviates from
standard reasoning only insofar as we need to handle the occurrences
of $\fal$ introduced by $\pa\E$. 
We counterbalance this
 introductions of $\fal$ by applying the rule $\fal \E$.

\begin{proposition}[Completeness] \label{proof:complete} For any
sequent $S= \Gamma \ergo \Delta $, if $S$ is derivable in $\SLL$, then
$\Gamma \ergo \Delta  $ is derivable in $\NLL$.
\end{proposition}
\begin{proof} We first remark that the restriction on $\NLL$ rule
applications can always be bypassed by renaming the variables in the
premises, so that the set of variables used in each premise is disjoint from the set of variables used in any other premise. Therefore, we ignore the term assignment also in the present proof. We proceed by
induction on the number of rule applications in the $\SLL$ derivation
of $S$.


If no rule is applied in the derivation, then either $S=\; A \ergo A$,
which is an initial sequent of $\NLL$, or $S= \fal \ergo \; $. Since in $\NLL$ we can
derive 
\[\infer[\fal\E]{\fal \ergo }{\fal\ergo \fal }\]the claim holds.

Suppose now that the statement holds for any $\SLL$ derivation
containing $n$ or less rule applications, we show the result for a
generic $\SLL$ derivation containing $n+1$ rule applications. We reason
by cases on the last rule applied in this $\SLL$ derivation.
\begin{itemize}

\item $\vcenter{\infer[\fal]{\Gamma \ergo \Delta , \fal }{\Gamma \ergo
\Delta}}$. Since the
derivation of the premise of the rule contain $n$ or less rule
applications, by inductive hypothesis we know that $\Gamma \ergo
\Delta $ is derivable in $\NLL$. Then by
the same application of the $\fal$ rule we can derive $\Gamma \ergo \Delta $ also in $\NLL$.

\item $ \vcenter{ \infer[\lo l]{ \Gamma , \Sigma , A \lo B \ergo
\Delta , \Theta }{\Gamma \ergo A, \Delta && \Sigma , B \ergo \Theta}
}$. Since the derivations of the premises of the rule contain $n$ or
less rule applications, by inductive hypothesis we know that $\Gamma
\ergo A, \Delta $ and $ \Sigma , B \ergo \Theta $ are derivable in
$\NLL$. 
By
\begin{footnotesize}
  \[\infer[\bot E]{\Gamma , \Sigma , A \lo B \ergo
\Delta , \Theta}{\infer[\lo \E]{\Gamma , \Sigma , A\lo B \ergo \Delta, \bot, \Theta 
      }{\infer[\lo \I]{\Sigma \ergo B\lo \bot,
        \Theta  }{\infer[\bot I]{\Sigma , B \ergo  \bot, \Theta}{\Sigma , B \ergo  \Theta} 
        }&& \infer[\lo \E]{\Gamma , A\lo B \ergo
        B, \Delta }{ A\lo B \ergo A\lo B &&
        \Gamma \ergo A, \Delta  }}}\]
\end{footnotesize}we can derive 
$ \Gamma , \Sigma , A\lo B \ergo  \Delta, \Theta $
in $\NLL$.

\item $ \vcenter{ \infer[\lo r]{ \Gamma \ergo A \lo B , \Delta
}{\Gamma , A \ergo B , \Delta }}$. Since the derivation of the
premise of the rule contain $n$ or less rule applications, by
inductive hypothesis we know that $\Gamma , A \ergo B , \Delta$ is
derivable in $\NLL$. By the rule application\[\infer[\lo \I]{ \Gamma \ergo A \lo B , \Delta 
}{\Gamma , A \ergo B , \Delta  }\]we can derive $ \Gamma \ergo A \lo B , \Delta$ in $\NLL$.

%

\item $ \vcenter{ 
\infer[\pa l ]{  \Gamma ,  \Sigma , A\pa B \ergo 
      \Delta , \Theta}{\Gamma , A \ergo  \Delta  && \Sigma , B \ergo  \Theta} }$. 
Since the derivations of the premises of the rule contain $n$ or
less rule applications, by inductive hypothesis we know that $\Gamma , A
\ergo  \Delta $ and $ \Sigma , B \ergo \Theta  $ are derivable in
$\NLL$. By the rule application
\begin{small}
  \[ \infer[\bot E]{\Gamma ,\Sigma, A\pa B \ergo  \Delta ,
      \Theta }{\infer[\pa \E]{ \Gamma ,\Sigma, A\pa B \ergo \fal, \Delta ,
      \Theta }{A\pa B \ergo A\pa B  && \Gamma , A \ergo \Delta
      && \Sigma , B \ergo \Theta }}\]
\end{small}we can derive $ \Gamma ,\Sigma,  A\pa B \ergo \Delta , \Theta $ in $\NLL$.

\item $ \vcenter{ \infer[\pa r ]{ \Gamma \ergo A\pa B , \Delta}{\Gamma
\ergo A , B , \Delta}}$. Since the derivation of the premise of the
rule contain $n$ or less rule applications, by inductive hypothesis we
know that $\Gamma  \ergo A, B, \Delta  $ is
derivable in $\NLL$.  By a rule application\[\infer[\pa \I ] { \Gamma
\ergo A\pa B , \Delta }{\Gamma \ergo A , B ,
\Delta }\]we can derive $ \Gamma \ergo A\pa B ,
\Delta$ in $\NLL$.

\item $\vcenter{\infer[\cut]{\Gamma , \Sigma \ergo \Delta ,
\Theta}{\Gamma \ergo A, \Delta && \Sigma , A \ergo \Theta}}$. Since
the derivations of the premises of the rule contain $n$ or less rule
applications, by inductive hypothesis we know that $\Gamma \ergo A,
\Delta $ and $ \Sigma , A \ergo \Theta $ are derivable in $\NLL$. By
\[\infer[\bot E]{\Gamma , \Sigma \ergo \Delta,  \Theta }{\infer[\lo \E]{\Gamma , \Sigma \ergo \Delta , \bot , \Theta  }{\infer[\lo \I ]{\Sigma \ergo A \lo \bot , \Theta  }{\infer[\bot I]{\Sigma , A \ergo \bot , \Theta}{\Sigma , A \ergo  \Theta} } &&
\Gamma \ergo A, \Delta  }}\] we can derive $
\Gamma , \Sigma \ergo \Delta , \Theta $ in
$\NLL$.
\end{itemize}\end{proof} 


\section{Subject Reduction}\label{sec:subject-reduction}

We are now going to prove that the reductions of $\lamp$ preserve the
typing of terms, property well-known under the name of Subject Reduction. 

We first need the concept of substitution, which in $\lamp$ is just a replacement of some variables, with no renaming involved.
\begin{definition}[Substitution]\mbox{}
 For any multiset of typed terms $\Delta$ and  terms $v_1 , \ldots ,
v_n$, we denote by $\Delta [v_1/x_1 \;\; \ldots \;\;v_n/x_n]$ the
simultaneous  replacement, for any $i\in \{1, \ldots ,n\}$, of all occurrences of
$x_i$ in all terms of $\Delta$ by $v_i$.
Given a substitution $[t/x]$, we refer to $t$ as
\emph{the substituting term} and to $x$ as \emph{the substituted
variable}.
\end{definition}

In order to prove Subject Reduction, we first prove that if we have a
variable $x:A$ in the context of a derivable sequent $\Gamma , x:A
\ergo \Delta $ and we can type a term $t:A$ by deriving a sequent
$\Sigma \ergo t:A , \Theta $, then we can derive the sequent $\Gamma ,
\Sigma \ergo \Delta[t/x] , \Theta $ effectively replacing the variable
$x:A$ with the term $t:A$. The proof-theoretical intuition is the
following: if $A$ is among the assumptions of a proof, and we have a
derivation $t$ of $A$, then we can use $t$ to derive $A$ directly
inside the proof.

\begin{lemma}[Substitution]\label{lem:derivation-substitution} If
$\Sigma \ergo t:A , \Theta $ and $\Gamma , x:A \ergo \Delta $ are
derivable in $\NLL$ and share no variable, then $\Gamma , \Sigma \ergo
\Delta[t/x] , \Theta $ is derivable in $\NLL$ as well.
\end{lemma}
\begin{proof} We proceed by induction on the number of rule
applications in the derivation of $\Gamma , x:A \ergo \Delta $.

If no rule is applied in the derivation, then $ \Gamma , x:A \ergo
\Delta $ is of the form $x: A \ergo x: A$. Therefore, $\Gamma , \Sigma \ergo
\Delta[t/x], \Theta $ is just $\Sigma \ergo t:A, \Theta $ and the claim
trivially holds.

Suppose now that the statement holds for any sequent which is
derivable with $n$ or less rule applications, we show the result for
any sequent $\Gamma , x:A \ergo \Delta $ which is derivable using
$n+1$ rule applications. We reason by cases on the last rule applied
in this derivation of $\Gamma , x:A \ergo \Delta $.
\begin{itemize}

\item $\vcenter{\infer[\fal\I]{\Gamma , x:A \ergo \Delta' , \termfal :
\fal }{\Gamma , x:A \ergo \Delta'}}$ where $\Delta = \Delta' ,
\termfal : \fal$.  By inductive hypothesis $\Gamma , \Sigma \ergo
\Delta'[t/x], \Theta $ is derivable.  By applying
\[\infer[\fal\I]{\Gamma , \Sigma \ergo \Delta'[t/x], \Theta,
\termfal:\fal }{\Gamma , \Sigma \ergo  \Delta'[t/x], \Theta}\] we
obtain a derivation of $\Gamma , \Sigma \ergo \Delta[t/x] , \Theta $,
which verifies the claim.

\item $ \vcenter{ \infer[\lo \E]{ \Gamma_1 , \Gamma_2 \ergo vw:C ,
\Delta_1, \Delta_2 }{\Gamma_1 \ergo v: B\lo C, \Delta_1 && \Gamma_2
\ergo w: B, \Delta_2} }$ where $\Delta = vw:C , \Delta_1,
\Delta_2$. Suppose that $x:A$ is contained in $\Gamma_1$ and hence
that $\Gamma _1 = \Gamma _1 ' ,x:A $ and that $\Gamma = \Gamma _1 ',
\Gamma _2 $. The case in which $x:A$ is contained in $\Gamma _2$ is
analogous. By inductive hypothesis $\Gamma_1' , \Sigma \ergo v[t/x]:
B\lo C, \Delta_1[t/x] , \Theta$ is derivable. Now, due to the type
assignment rules of $\NLL$ and since $x:A$ is
contained in $\Gamma_1$, we have that $x$ cannot occur in $w$.  By hypothesis, $t$ and $\Theta$ share no variable  with $\Gamma_{2}, w,\Delta_{2}$.
 Hence,
by applying
\[ \infer[\lo \E]{ \Gamma_1 , \Sigma , \Gamma_2 \ergo (v[t/x])w:C ,
\Delta_1[t/x],\Theta , \Delta_2 }{\Gamma_1, \Sigma \ergo v[t/x]: B\lo
C, \Delta_1[t/x],\Theta && \Gamma_2 \ergo w: B, \Delta_2}\]we obtain a
derivation of $\Gamma , \Sigma \ergo (vw)[t/x]:C ,
\Delta[t/x],\Theta$, which verifies the claim.

\item $ \vcenter{ \infer[\lo \I]{ \Gamma \ergo \send{y} s : B \lo C ,
\Delta' }{\Gamma , y: B \ergo s: C , \Delta' }}$ where $\Delta = \send{y} s : B \lo C , \Delta'$. Since $x:A$ must occur in $\Gamma$ in the
conclusion of the rule application, we know that $x\neq y$. By
inductive hypothesis, $\Gamma , y: B , \Sigma \ergo s[t/x]: C ,
\Delta'[t/x] , \Theta$ is derivable. Hence, since by hypothesis $y$ does not occur
in $t$, by
applying\[\infer[\lo \I]{ \Gamma , \Sigma \ergo \send{y} (s[t/x]) : B
\lo C , \Delta'[t/x],\Theta }{\Gamma , y: B , \Sigma \ergo s[t/x]: C ,
\Delta'[t/x] ,\Theta }\]we obtain a derivation of $\Gamma , \Sigma
\ergo (\send{y} s)[t/x] : B \lo C , \Delta'[t/x],\Theta$, which
verifies the claim.

\item $ \vcenter{\infer[\pa \E]{ \Gamma_1 , \Gamma _2 ,\Gamma_3 \ergo
\distp{y}{z} t: \fal , \Delta_1 , \Delta _2 , \Delta _3 }{\Gamma_1
\ergo s: B\pa C , \Delta_1 && \Gamma _2 , y: B \ergo \Delta _2 &&
\Gamma_3 , z: C \ergo \Delta _3}}$ where $\Delta = \distp{y}{z} t:
\fal , \Delta_1 , \Delta _2 , \Delta _3$. Suppose that $x:A$ is
contained in $\Gamma_1$ and hence that $\Gamma _1 = \Gamma _1 ' ,x:A $
and that $\Gamma = \Gamma _1 ', \Gamma _2, \Gamma _3 $.  The case in
which $x:A$ is contained in $\Gamma _2$ or in $\Gamma _3$ are
analogous.  By inductive hypothesis $\Gamma_1 , \Sigma \ergo s[t/x]:
B\pa C , \Delta_1[t/x] , \Theta$ is derivable. By hypothesis, $t, \Theta$ and $\Gamma_{2}, y: B, \Delta_{2}$ and  $\Gamma_{3}, z: C, \Delta_{3}$ share no variable.
 Therefore, by applying
\[\infer[\pa \E]{ \Gamma_1, \Sigma , \Gamma _2 ,\Gamma_3 \ergo \distp{y}{z} (s[t/x]): \fal , \Delta_1[t/x] ,\Theta , \Delta _2 , \Delta _3
}{\Gamma_1 , \Sigma\ergo s[t/x]: B\pa C , \Delta_1[t/x],\Theta &&
\Gamma _2 , y: B \ergo \Delta _2 && \Gamma_3 , z: C \ergo \Delta _3}\]
we obtain a derivation of $ \Gamma, \Sigma \ergo (\distp{y}{z}
s)[t/x]: \fal , \Delta[t/x] ,\Theta$, which
verifies the claim.

\item $ \vcenter{ \infer[\pa \I ]{ \Gamma , x:A \ergo v\p w : B \pa C
, \Delta'}{\Gamma , x:A \ergo v: B , w: C, \Delta'}}$ where $\Delta =
v\p w : B \pa C , \Delta'$.  By inductive hypothesis $\Gamma , \Sigma
\ergo v[t/x]: B , w[t/x]: C, \Delta'[t/x] , \Theta$ is derivable. By
applying
\[\infer[\pa \I ]{ \Gamma, \Sigma \ergo (v\p w)[t/x] : B \pa C ,
\Delta'[t/x],\Theta }{\Gamma , \Sigma \ergo v[t/x]: B , w[t/x]: C,
\Delta'[t/x],\Theta }\]we obtain a derivation of $ \Gamma, \Sigma
\ergo (v\p w)[t/x] : B \pa C , \Delta'[t/x],\Theta$, which verifies
the claim.
\end{itemize}
\end{proof}

The  proof  of the Subject Reduction is now quite standard.

\begin{theorem}[Subject Reduction] \label{thm-subred}
 Assume there is a $\NLL$ derivation of $ \Pi \ergo  t_1:F_1 ,
\ldots , t_n:F_n, t_{n+1}, \dots, t_{m} $. If $t_1 \p \ldots \p t_m\ \red t_1' \p \ldots \p
t_m'$, then there is a derivation of $ \Pi \ergo \Phi, t_1':F_1 , \ldots ,
t_n':F_n, t_{n+1}', \ldots, t_{m}' $.

\end{theorem}
\begin{proof}

We prove the statement by induction on the number of rule
applications in the derivation of $ \Pi \ergo t_1:F_1 , \ldots ,
t_n:F_n, t_{n+1}, \dots, t_{m} $.

If no rule is applied in the derivation, then  $ \Pi \ergo t_1:F_1 , \ldots ,
t_n:F_n, t_{n+1}, \dots, t_{m} $ is of the form $ x: A \ergo x: A$  and the claim holds. Since no reduction applies to any term in the right-hand side of the sequents, the claim trivially holds.

Suppose now that the statement holds for any $\NLL$ derivation
containing $m$ or less rule applications, we show the result for a
generic $\NLL$ derivation of $ \Pi \ergo t_1:F_1 , \ldots , t_n:F_n, t_{n+1}, \dots, t_{m}
$ containing $m+1$ rule applications. We reason by cases on the last
rule applied in this $\NLL$ derivation.

\begin{itemize}
\item $ \vcenter{ \infer[\lo \E]{ \Gamma , \Sigma \ergo \sendm{x}{
v}{w}:B , \Delta , \Theta }{\Gamma \ergo \send{x} v: A\lo B, \Delta &&
\Sigma \ergo w: A, \Theta} }$ where \[t_1 \p \ldots \p t_m\quad =\quad
t_1 \p \ldots \p \sendm{x}{v}{w} \p \ldots \p t_m \]and
\[t_1' \p \ldots \p t_m' \quad = \quad (t_1 \p \ldots \p v \p \ldots
\p t_m )[w/x]\]
since $x$ occurs only once.

\begin{itemize}
\item If the last rule applied above $\Gamma \ergo \send{x}
v: A\lo B, \Delta$ is
\[\infer[\lo \I]{\Gamma \ergo \send{x} v: A\lo B, \Delta}{\Gamma ,
x:A \ergo v: B, \Delta}\]

then, by the restriction on $\NLL$ rule applications, $x$ occurs
either in $v$ or in $\Delta$, by Lemma~\ref{lem:derivation-substitution}
applied to the sequents $\Sigma \ergo w: A, \Theta$ and $\Gamma , x:A
\ergo v: B, \Delta$ we obtain that $\Gamma , \Sigma \ergo v[w/x]: B,
\Delta[w/x], \Theta$ is derivable and hence that we have a derivation
of $ \Pi \ergo t_1':F_1 , \ldots , t_n':F_n, t_{n+1}', \dots, t_{m}' $.
\item  Otherwise, we have\begin{small}
\[\infer[\lo \E]{ \Gamma , \Sigma \ergo \sendm{x}{v}{w}:B , \Delta , \Theta }{\infer[\rho]{\Gamma \ergo \send{x} v :
A\lo B, \Delta}{\Phi_1 \ergo \send{x} v : A\lo B, \Psi_1
&&\ldots&& \Phi_p \ergo \Psi_p} && \Sigma , w:A \ergo \Theta}\]
\end{small}where $ \send{x}v :A \lo B$ only occurs in one of the
premises of $\rho$ -- we display it in the first premise without loss of
generality. 

If we construct the following derivation
\begin{small}
\[ \infer[\rho]{\Gamma , \Sigma \ergo \sendm{x}{v}{ w}:B , \Delta, \Theta
}{\infer[\lo \E]{ \Phi_1 , \Sigma \ergo \sendm{x}{v}{w}:B , \Psi_1 ,
\Theta }{ \Phi_1 \ergo \send{x} v :A \lo B, \Psi_1 && \Sigma , w:A
\ergo \Theta}&&\ldots&& \Phi_p \ergo \Psi_p}
\]
\end{small} we know without loss of generality that \[ \sendm{x}{v}{w}:B , \Psi_1 , \Theta \quad
= \quad t_1:F_1 , \ldots , t_j:F_j, t_{n+1}, \dots, t_{i}\] and that \[t_1 \p \ldots \p t_j \p t_{n+1}\p \dots\p t_{l}
\quad \red \quad t_1' \p \ldots \p t_j'\p t_{n+1}'\p \dots\p t_{i}'\] because $x$ occurs either
in $v$ or in $\Psi_1$. By inductive hypothesis, there is a derivation
of \[ \Phi_1 , \Sigma \ergo t_1':F_1 , \ldots , t_j':F_j, t_{n+1}', \dots, t_{i}' \quad\] 
\[=\quad
\Phi_1 , \Sigma \ergo v[w/x]:B , \Psi_1[w/x] , \Theta\]By assumption
we have \[ \Gamma , \Sigma \ergo \sendm{x}{v}{w}:B , \Delta, \Theta
\quad = \quad \Pi \ergo t_1:F_1 , \ldots , t_n:F_n, t_{n+1}, \ldots, t_{m} \]Therefore, by
inspection of the rules of $\NLL$, it is easy to see that by the rule application
\[ \infer[\rho]{\Gamma , \Sigma \ergo v[w/x]:B , \Delta[w/x] , \Theta
}{\Phi_1 , \Sigma \ergo v[w/x]:B , \Psi_1[w/x] , \Theta &&\ldots&&
\Phi_p \ergo \Psi_p}
\]we obtain a derivation of the sequent\[ \Gamma , \Sigma \ergo
v[w/x]:B , \Delta[w/x] , \Theta\quad =\quad \Gamma , \Sigma \ergo
t_1':F_1 , \ldots , t_n':F_n, t_{n+1}', \ldots, t_{m}' \]
\end{itemize}

\item $ \vcenter{ \infer[\pa \E]{ \Gamma , \Sigma_1 , \Sigma_2 \ergo
\distp{x}{y}(v\p w) : \fal , \Delta , \Theta_1, \Theta_2 }{\Gamma
\ergo v\p w : A\pa B, \Delta && \Sigma_1 , x:A \ergo \Theta_1 &&
\Sigma_2 , y:B \ergo \Theta_2} }$ where \[t_1 \p \ldots \p t_m\quad
=\quad t_1 \p \ldots \p \distp{x}{y} (v\p w)  \p \ldots \p
t_m \]and
\[t_1' \p \ldots \p t_m' \quad = \quad (t_1 \p \ldots \p \termfal \p \ldots
\p t_m )[v/x\;\;w/y]\]

\begin{itemize}
\item If the last rule applied above $\Gamma \ergo v\p w : A\pa B,
\Delta$ is
\[\infer[\pa \I]{\Gamma \ergo v\p w : A\pa B, \Delta }{\Gamma \ergo v:
A,w:B, \Delta }\]

By the application of Lemma~\ref{lem:derivation-substitution} to the sequents
$\Gamma \ergo v: A,w:B , \Delta$ and $\Sigma_1 , x:A \ergo \Theta_1$
we obtain that $\Gamma , \Sigma _1 \ergo w:B , \Delta , \Theta_1[v/x]$
is derivable. Furthermore, by applying
Lemma~\ref{lem:derivation-substitution} to the sequents $\Gamma , \Sigma _1
\ergo w:B , \Delta , \Theta_1[v/x]$ and $\Sigma_2 , y:B \ergo
\Theta_2$ we obtain that $\Gamma , \Sigma _1, \Sigma_2 \ergo \Delta ,
\Theta_1[v/x], \Theta_2[w/y]$. Since
$x$ occurs in $\Theta _1$ and $y$ occurs in $\Theta _2$, we can
construct a derivation of $ \Pi \ergo t_1':F_1 , \ldots , t_n':F_n, t_{n+1}', \ldots, t_{m}' $
by applying $\fal\I$.

\item Otherwise, we have\begin{scriptsize}
\[\infer[\pa \E]{ \Gamma , \Sigma \ergo \distp{x}{y} (v \p w) : \fal ,
\Delta , \Theta_1 , \Theta_2 }{\infer[\rho]{\Gamma \ergo v \p w : A\pa B,
\Delta}{\Phi_1 \ergo v \p w : A\pa B, \Psi_1 &&\ldots&& \Phi_p \ergo
\Psi_p} & \Sigma_1 , x:A \ergo \Theta_1 & \Sigma_2 , y:B \ergo
\Theta_2 }\]\end{scriptsize}where $v \p w : A\pa B$ only occurs in one
of the premises of $\rho$ -- we display it in the first premise without
loss of generality.

If we construct the following derivation
\begin{scriptsize}
\[ \infer[\rho]{\Gamma , \Sigma \ergo  \distp{x}{y} (v \p w) : \fal , \Delta,
\Theta_1 , \Theta_2 }{\infer[\pa \E]{ \Phi_1 , \Sigma \ergo \distp{x}{y} (v \p w) : \fal , \Psi_1 , \Theta_1 , \Theta_2 }{ \Phi_1 \ergo v \p w : A\pa
B, \Psi_1 && \Sigma_1 , x:A \ergo \Theta_1 & \Sigma_2 , y:B \ergo
\Theta_2 }&&\ldots&& \Phi_p \ergo
\Psi_p}
\]\end{scriptsize} we know without loss of generality that \[ \distp{x}{y} (v \p w) : \fal ,
\Psi_1 , \Theta_1 , \Theta_2 \quad = \quad t_1:F_1 , \ldots , t_j:F_j, t_{n+1}, \ldots, t_{i}\]
and that \[t_1 \p \ldots \p t_j\p t_{n+1}\p \ldots\p t_{i} \quad \red \quad t_1' \p \ldots \p
t_j'\p t_{n+1}'\p \ldots\p t_{i}'\] because $x$ and $y$ occur in $\Theta_1 $ and $ \Theta_2$
respectively. By inductive hypothesis, there is a derivation of \[
\Phi_1 , \Sigma \ergo t_1':F_1 , \ldots ,
t_j':F_j, t_{n+1}', \ldots, t_{i}' \quad =\quad \Phi_1 , \Sigma \ergo \termfal:\fal , \Psi_1 ,
\Theta_1[v/x] , \Theta_2[w/y]\] By assumption\[ \Gamma , \Sigma \ergo
\distp{x}{y} (v \p w) : \fal , \Delta , \Theta_1 , \Theta_2 \quad = \quad
\Pi \ergo t_1:F_1 , \ldots , t_n:F_n, t_{n+1}, \ldots, t_{m} \]Therefore, by the rule application
\begin{small}
\[ \infer[\rho]{\Gamma , \Sigma \ergo \termfal:\fal, \Delta, \Theta_1[v/x] , \Theta_2[w/y] }{\Phi_1 , \Sigma \ergo \termfal:\fal , \Psi_1 , \Theta_1[v/x] , \Theta_2[w/y] &&\ldots&& \Phi_p \ergo \Psi_p}
\]\end{small}we have a derivation
of the sequent\[ \Gamma , \Sigma \ergo \termfal:\fal,
\Delta,  \Theta_1[v/x] , \Theta_2[w/y]\quad =\quad \Gamma , \Sigma \ergo
t_1':F_1 , \ldots , t_n':F_n, t_{n+1}', \ldots, t_{m}'\]
\end{itemize}




\item In all other cases, the term $\sendm{x}{v}{w} $ (or $\distp{x}{y} v$) that triggered the reduction $t_1 \p \ldots \p t_m
\red t_1' \p \ldots \p t_m'$ already occurs in the premises of the
last rule $\rho$ applied in the $\NLL$ derivation of $ \Pi \ergo t_1:F_1
, \ldots , t_n:F_n, t_{n+1}, \ldots, t_{m} $.


According to the typing rules, no variable can occur in different
premises of the same rule application. Therefore, both the term
$\sendm{x}{v}{w}$ (respectively $\distp{y}{z}v$) and the variable $x$
(respectively $y$ and $z$) occur inside the same premise of the last
rule application \[\infer[\rho] {\Sigma \ergo \Delta, \calc_1[u_1]:F_j
, \ldots , \calc_p[u_p]:F_n}{\Gamma \ergo u_1:G_1 , \ldots , u_p:G_p
&&\ldots&& \Phi \ergo \Psi}\]in the derivation of\[\Sigma \ergo
\Delta, \calc_1[u_1]:F_j , \ldots , \calc_p[u_p]:F_n \;\; =\;\; \Pi
\ergo t_1:F_1 , \ldots , t_n:F_n, t_{n+1}, \ldots, t_{m}\] Without
loss of generality, we assume that $x$ (or $y$ and $z$) occurs in the
premise \[\Gamma \ergo u_1:G_1 , \ldots , u_p:G_p\]

By reducing the redex triggered by $\sendm{x}{v}{w}$ 
or $\distp{x}{y}v$ in the term $u_1 \p \ldots \p u_p$ we have $u_1 \p \ldots \p u_p
\red u_1' \p \ldots \p u_p'$. By inductive hypothesis, there is a
derivation of $ \Gamma \ergo u_1':G_1 , \ldots , u_p':G_p $. Now, by
inspection of $\NLL$ typing rules, we can easily see that $t_1\p
\ldots \p t_m$ is of the form $\cald[\; \calc_1[u_1] \p \ldots \p
\calc_p[u_p]\; ]$ and by assumption
$$t_1'\p
\ldots \p t_m' = \cald[\; \calc_1[u_1'] \p \ldots \p
\calc_p[u_p']\; ]$$


Therefore, by inspection of the rules of $\NLL$, it
is easy to see that by the application of $\rho$
\[\infer[\rho] {\Sigma \ergo \Delta , \calc_1[u_1'] :F_j, \ldots ,
\calc_p[u_p']:F_n}{\Gamma \ergo u_1':G_1 , \ldots , u_p':G_p
&&\ldots&& \Phi \ergo \Psi}\] to the root of the derivation
of $ \Gamma \ergo u_1':G_1 , \ldots , u_p':G_p $, we obtain a
derivation of \[ \Sigma \ergo \Delta , \calc_1[u_1'] :F_j, \ldots ,
\calc_p[u_p']:F_n \;=\;\Pi \ergo t_1':F_1 , \ldots , t_n':F_n, t_{n+1}', \ldots, t_{m}' \]
\end{itemize}
\end{proof}

\section{Progress}\label{sec:progress}

 We show now that $\lamp$ is deadlock-free: if a process contains a potential communication,
then the term is not normal and the communication will be carried
out. In $\lamp$, potential communications are represented by
subterms of the form $\sendm{x}{w}{v}$
or
$\distp{x}{y}v$. Indeed, the presence of one of such subterms in a
process means that the process can use a communication channel $x$ --
or a pair of channels $x$ and $y$ -- to transmit a message $v$. Technically, then, we need to show that if such a subterm occurs in a term, there is also a receiver, hence the term is not normal. We start by showing some properties of the distribution of variables  inside $\NLL$ sequents and inside typed $\lamp$-terms 
\begin{proposition}[Linearity of Channels]\label{prop:linchan} Assume $\Gamma\ergo
\Delta $ is derivable in $\NLL$ and let $x$ be any variable. Then:
\begin{itemize}
\item  if $x$ occurs in $\Gamma $, then $x$ occurs  exactly  once in $\Delta$ and $\overline{x}$ does not occur in $\Delta$;
\item if $x $ does not occur in $\Gamma$, but occurs in $\Delta$,  then 
$x$ occurs exactly once as $x$ and exactly once as $\overline{x}$ in $\Delta$.
\end{itemize}
\end{proposition}
\begin{proof}
By straightforward induction on the length of the derivation of $\Gamma\ergo
\Delta $.
\end{proof}

We recall that by Definition \ref{def-context} a simple context
$\calc[\;]$ is a process which is not a direct parallel composition of
simpler processes. A crucial property of simple contexts that we are
going to obtain now is the following: if $\calc[\sendm{z}{v}{u}]$ is
typable and $\calc[\;]$ is simple, then the variable $z$ cannot occur
in $\calc[\;]$.  If such a configuration were possible, we might type
$\lamp$-terms like
$$z (\sendm{z}{v}{u})$$
and we would have to choose between allowing components of the same
process reduction to communicate
\[
\begin{array}{ccc} &\vspace{-5pt}?& \\ z ( \sendm{z}{v}{u})\ & \mapsto
& \ uv
\end{array}\] or tolerating a deadlock. Fortunately, the next
proposition rules out this scenery.

\begin{proposition}\label{prop-nofreaks}
Suppose $\Phi \ergo \calc[u]: F, \Pi$ or $\Phi \ergo \calc[u], \Pi$ is derivable in $\NLL$. Then, if $\calc[\;]$ is simple and the variable  $\overline{z}$ occurs in $u$, then $z$ does not occur in $\calc[\;]$.
\end{proposition}

\begin{proof} We prove the statement by induction on the number of
rule applications in the derivation of $\Phi\ergo \calc[u]: F, \Pi$ or $\Phi\ergo \calc[u], \Pi$. 

If $\calc[\;]$ is empty, we are done.

If no rule is applied in the derivation, then $\Phi \ergo \calc[u]:F,
\Pi $ or $\Phi\ergo \calc[u], \Pi$ is  of the form $ x: A \ergo x: A$, thus $u=x$ and $\calc[\;]$ is empty, therefore  the claim trivially holds.

Suppose now that the statement holds for any $\NLL$ derivation
containing $m$ or less rule applications, we show the result for a
generic $\NLL$ derivation containing $m+1$ rule applications. We
reason by cases on the last rule applied in this $\NLL$ derivation. We may assume that the term $\calc[u]$ does not occur in any premise of this last rule, otherwise we just apply the inductive hypothesis to any premise containing  $\calc[u]$ and obtain the thesis.
\begin{itemize}

\item $\vcenter{\infer[\fal\I]{\Gamma \ergo \Delta , \termfal : \fal
}{\Gamma \ergo \Delta}}$. Since $\calc[u]$ occurs in $\Delta$, this case is ruled out by our assumption that  $\calc[u]$ does not occur in $\Gamma \ergo \Delta$.

\item $\vcenter{\infer[\fal\E]{\Gamma \ergo \Delta , \abort(t)
}{\Gamma \ergo \Delta, t: \fal}}$. By our assumption, it must be the case that $\calc[u]=\abort(t)$. If $\calc[\;]=\abort([\;])$, surely $z$ does not occur in $\calc[\;]$. If  $\calc[\;]=\abort(\cald[\;])$, with $t=\cald[u]$, then by inductive hypothesis, $z$ does not occur in $\cald[\;]$, thus it does not occur in $\calc[\;]=\abort(\cald[\;])$ either.


\item $ \vcenter{ \infer[\lo \E]{ \Gamma , \Sigma \ergo ts:B , \Delta
, \Theta }{\Gamma \ergo t: A\lo B, \Delta && \Sigma \ergo s: A,
\Theta} }$.  By our assumption, it must be the case that $\calc[u]=t s$. 
Therefore, $t=\cald[u]$ or $s=\cald[u]$, with respectively $\calc[\;]=\cald[\;]\,s$ and $\calc[\;]=t\, \cald[\;]$. By inductive hypothesis, $z$ does not occur in $\cald[\;]$. Moreover, since the premises of the typing rule share no variable, $z$ cannot occur respectively in $s$ and $t$. Therefore, $z$ does not occur in $\calc[\;]$.

\item $ \vcenter{ \infer[\lo \I]{ \Gamma \ergo \send{x} t : A \lo B ,
\Delta }{\Gamma , x: A \ergo t: B , \Delta }}$. By our assumption, it
must be the case that $\calc[u]=\send{x} t$. Therefore, $t=\cald[u]$
and $\calc[\;]= \send{x} \cald[\;]$. By inductive hypothesis, $z$ does
not occur in $\cald[\;]$. Moreover, $z\neq x$, otherwise
$\overline{z}$ would occur twice in $\send{x} t$, contradicting
Proposition \ref{prop:linchan}. Therefore, $z$ does not occur in
$\calc[\;]=\send{x} \cald[\;]$.

%
%
%

\item $ \vcenter{\infer[\pa \E]{ \Gamma ,\Sigma_1 ,\Sigma_2 \ergo \distp{x}{y} t:
\fal , \Delta , \Theta_1 , \Theta_2 }{\Gamma \ergo t: A\pa B , \Delta
&& \Sigma _1 , x: A \ergo \Theta _1 && \Sigma _2 , y: B \ergo  \Theta _2 } }$.  By our assumption, it must be the case that $\calc[u]=\distp{x}{y} t$. Therefore,  $t=\cald[u]$ and $\calc[\;]= \distp{x}{y} \cald[\;]$. By inductive hypothesis, $z$ does not occur in $\cald[\;]$. Moreover, $z\neq x, y$, otherwise $\overline{z}$ would occur twice in $\distp{x}{y} t$, contradicting Proposition \ref{prop:linchan}. Therefore $z$ does not occur in $\calc[\;]= \distp{x}{y} \cald[\;]$.

\item $ \vcenter{ \infer[\pa \I ]{ \Gamma \ergo t\p s : A \pa B ,
\Delta}{\Gamma \ergo t: A , s: B, \Delta}}$. We claim that this case is not possible. Indeed, if it were, by our assumption we would have $\calc[u]= t\p s$. Therefore  $t=\cald[u]$ or $s=\cald[u]$, with respectively $\calc[\;]= \cald[\;]\p s$ and $\calc[\;]= t\p \cald[\;]$. Since by Definition \ref{def-context} the context $\calc[\;]$ would not be  simple, we would have a contradiction.

\end{itemize}
\end{proof}
It is now easy to show that if a communication is possible, it will
be carried out: if a typable $\lamp$-term contains output channels ready to
transmit messages, then it can always perform the communication,
because there is always a suitable receiver.
 
\begin{theorem}[Progress]\mbox{} 
\begin{enumerate}
\item Suppose 
$$\Gamma \ergo \calc[\sendm{x}{u}{v}]: F, t_{1}: F_{1}, \ldots, t_{n}: F_{n}, t_{n+1}, \ldots, t_{m} $$
is derivable in $\NLL$. Then the term  $\calc[\sendm{x}{u}{v}]\p t_{1}\p \ldots\p t_{m}$ is not normal.
\item Suppose 
$$\Gamma \ergo \calc[\distp{x}{y} u]: F, t_{1}: F_{1}, \ldots, t_{n}: F_{n}, t_{n+1}, \ldots, t_{m} $$
is derivable in $\NLL$. Then the term  $\calc[\distp{x}{y} u]\p t_{1}\p \ldots\p t_{m}$ is not normal.
\end{enumerate}
\end{theorem}
\begin{proof}\mbox{}

1.   By Proposition \ref{prop:linchan}, $\overline{x}$ and $x$ must occur exactly once in the term $\calc[\sendm{x}{u}{v}]\p t_{1}\p \ldots\p t_{m}$. If $x$ occurs in some $t_{i}$, then a reduction is possible and we are done. We assume therefore that $x$ occurs in $\calc[\sendm{x}{u}{v}]$.  We also assume  $\calc[\;]$ is empty, otherwise $x$ occurs in $u$, hence $\sendm{x}{u}{v}=(\lambda x\, u)v$, which is not normal and we are done.

We claim now that  $\calc[\sendm{x}{u}{v}]$  has a subterm of the form $\cale[\sendm{x}{u}{v}] \p \cald[x]$ or $  \cald[x] \p \cale[\sendm{x}{u}{v}]$, which implies a communication reduction  is possible, allowing us to conclude the proof. We prove our claim by induction on $\calc[\sendm{x}{u}{v}]$.  

If $\calc[\sendm{x}{u}{v}]=\send{z}\calc'[\sendm{x}{u}{v}]$ or $\calc[\sendm{x}{u}{v}]=\distp{y}{z}{\calc'[\sendm{x}{u}{v}]}$ or $\calc[\sendm{x}{u}{v}]= \abort(\calc'[\sendm{x}{u}{v}])$, then we apply the induction hypothesis on $\calc'[\sendm{x}{u}{v}]$, and we are done. 

If $\calc[\sendm{x}{u}{v}]= w\, \calc'[\sendm{x}{u}{v}]$ or $\calc[\sendm{x}{u}{v}]=  \calc'[\sendm{x}{u}{v}]\, w$, then by Proposition \ref{prop-nofreaks}, $x$ must occur in $\calc'[\sendm{x}{u}{v}]$, hence we apply the induction hypothesis on $\calc'[\sendm{x}{u}{v}]$, and we are done. 

If $\calc[\sendm{x}{u}{v}]= \calc'[\sendm{x}{u}{v}] \p w$ or $\calc[\sendm{x}{u}{v}]= w\p \calc'[\sendm{x}{u}{v}]$, then if $x$ occurs in $w$, we are done. Otherwise, we apply the induction hypothesis on $\calc'[\sendm{x}{u}{v}]$, and concluding the proof of the claim.

2. By Proposition \ref{prop:linchan}, $\overline{x}$ and $x$ must occur both exactly once in the term $\calc[\distp{x}{y} u]\p t_{1}\p \ldots\p t_{m}$ and so do $\overline{y}$ and $y$.  By inspection of the inference rules of $\NLL$ it is easy to see that $x, y$ cannot occur in $u$, therefore they occur in some $t_{i}$. Thus, a reduction is possible.
\end{proof}

\section{The Subformula Property}\label{sec:subformula}

We show in this section that each normal $\lamp$-term corresponds to an analytic
derivation. In other words, the type of each subterm of any
normal $\lamp$-term $t$ is either a subformula of the type of $t$ or a
subformula of the types of the  variables of $t$. This Subformula Property guarantees that,  proof-theoretically, the reduction rules of $\lamp$ give rise to a complete detour removal procedure, hence we can conclude they are also satisfactory from the logical point of view. Indeed, without this property is not possible to relate normalization to cut-elimination, as a normal proof without the subformula property cannot be translated in to a cut-free proof to begin with. Our multiple-conclusion natural deduction system can then be considered as a well-behaved alternative to sequent calculus for  and proof-nets for multiplicative linear logic.

We start off by defining what it means for a derivation to be in normal form.
\begin{definition}[Normal form] 
An $\NLL$ derivation
of a sequent $\Gamma  \ergo t_1: T_1 , \ldots ,
t_n:T_n, t_{n+1}, \ldots, t_{m}$ is in normal form if the term $t_1\p \ldots \p t_m$ is in
normal form.
\end{definition}

We then recall the notion of stack~\cite{Krivine}. From a logical
perspective, a stack represents a series of elimination
rules.\begin{definition}[Stack]\label{def:stack} A \emph{stack} is a
possibly empty sequence \mbox{$\sigma = \xi_{1}\xi_{2}\ldots \xi_{n}
$} such that for every $ 1\leq i\leq n$, $\xi_{i}=t$, with $t$ 
term of $\NLL$.  
 If $t$ is a proof term, $t\, \sigma$ denotes as
usual the term $(((t\, \xi_{1})\,\xi_{2})\ldots \xi_{n})$.
 \end{definition}

 Before proving the Subformula Property, we need to
establish two auxiliary results concerning the shape of terms. The
first of these two results guarantees that the expected connection
between the shape of a term with its type holds.

\begin{proposition}[Type coherence]\label{prop:coherence} If $\Gamma
\ergo t:F, \Delta $ is derivable in
$\NLL$ and $t$ is a value, then one of the following holds:
\begin{itemize}
\item $t= \send{x}  t$ and $F= A \lo B  $
\item $t= s\p t$ and $F= A \pa B $
\end{itemize}for some types  $A $ and $B$.
\end{proposition}
\begin{proof} We prove the statement by induction on the number of
rule applications in the derivation of $\Gamma \ergo t:F, \Delta $.

If no rule is applied in the derivation, then $\Gamma \ergo t:F,
\Delta $ is  of the form $ x: A \ergo x: A$  and the claim trivially holds.

Suppose now that the statement holds for any $\NLL$ derivation
containing $m$ or less rule applications, we show the result for a
generic $\NLL$ derivation containing $m+1$ rule applications. We
reason by cases on the last rule applied in this $\NLL$ derivation.
\begin{itemize}

\item $\vcenter{\infer[\fal\I]{\Gamma \ergo \Delta , \termfal : \fal
}{\Gamma \ergo \Delta}}$. By inductive hypothesis the claim holds for
any value occurring in $ \Gamma \ergo \Delta$. Since $\termfal$ is not
a value, the claim holds also for any value occurring in $\Gamma \ergo
\Delta , \termfal : \fal $.

\item $\vcenter{\infer[\fal\E]{\Gamma \ergo \Delta , \abort(t)
}{\Gamma \ergo \Delta, t: \fal}}$. By inductive hypothesis the claim holds for
any value occurring in $ \Gamma \ergo \Delta$. Since $\abort(t)$ is not
a value, the claim holds also for any value occurring in $\Gamma \ergo
\Delta , \abort(t) $.


\item $ \vcenter{ \infer[\lo \E]{ \Gamma , \Sigma \ergo ts:B , \Delta
, \Theta }{\Gamma \ergo t: A\lo B, \Delta && \Sigma \ergo s: A,
\Theta} }$. By inductive hypothesis the claim holds for
any value occurring in $ \Gamma \ergo t: A\lo B, \Delta $ and $
\Sigma \ergo s: A , \Theta$. Since  $ts$ is not a value, the claim holds for all values in $\Gamma , \Sigma \ergo ts:B , \Delta
, \Theta $ as well.

\item $ \vcenter{ \infer[\lo \I]{ \Gamma \ergo \send{x} t : A \lo B ,
\Delta }{\Gamma , x: A \ergo t: B , \Delta }}$. By inductive
hypothesis the claim holds for any value occurring in $\Gamma , A
\ergo B , \Delta$. Since the value $\send{x} t : A \lo B$ has the
required type, the claim holds for all values in $ \Gamma \ergo \send
x\, t : A \lo B , \Delta $ as well.

%
 
\item $ \vcenter{\infer[\pa \E]{ \Gamma ,\Sigma_1 ,\Sigma_2 \ergo \distp{x}{y} t:
\fal , \Delta , \Theta_1 , \Theta_2 }{\Gamma \ergo t: A\pa B , \Delta
&& \Sigma _1 , x: A \ergo \Theta _1 && \Sigma _2 , y: B \ergo  \Theta _2 } }$.  \\By inductive hypothesis the claim holds for any
value occurring in $\Gamma \ergo t: A\pa B , \Delta $ and $ \Sigma _1
, x: A \ergo \Theta _1 $ and $ \Sigma _2 , y: B \ergo \Theta _2$.
Since $\distp{x}{y} t: \fal$ is not a value, the claim holds for all
values in $ \Gamma ,\Sigma_1 ,\Sigma_2 \ergo \distp{x}{y} t: \fal ,
\Delta , \Theta_1 , \Theta_2$ as well.

\item $ \vcenter{ \infer[\pa \I ]{ \Gamma \ergo t\p s : A \pa B ,
\Delta}{\Gamma \ergo t: A , s: B, \Delta}}$. By inductive hypothesis
the claim holds for any value occurring in $\Gamma \ergo t: A , s: B,
\Delta$. Since the value $t\p s : A \pa B$ has the required type, the
claim holds for $\Gamma \ergo t\p s : A \pa B , \Delta$ as well.
\end{itemize}
\end{proof}

 The second auxiliary result establishes that any normal
$\lamp$-term -- the type of which is not $ \fal$ -- is either a value
or a sequence of operations that cannot be carried out.

\begin{lemma}[The shape of normal terms]\label{lem:stackform} If there
is a derivation in $\NLL$ of $\Gamma \ergo u: F, \Pi$, where $u$ is in normal form,  is not a value and $F\neq \fal$, then $u = y\sigma$.
\end{lemma}
\begin{proof} We prove the statement by induction on the number of
rule applications in the derivation of $\Gamma \ergo u: F, \Pi$.

If no rule is applied in the derivation, then $\Gamma \ergo u: F, \Pi$ is of the form $ x: A
\ergo x: A$, thus
the claim holds since $u=x$. 

Suppose now that the statement holds for any $\NLL$ derivation
containing $m$ or less rule applications, we show the result for a
generic $\NLL$ derivation containing $m+1$ rule applications. We
reason by cases on the last rule applied in this $\NLL$ derivation.
\begin{itemize}

\item $\vcenter{\infer[\fal\I]{\Gamma \ergo \Delta , \termfal : \fal
}{\Gamma \ergo \Delta}}$. By inductive hypothesis the claim holds for
any term occurring in $ \Gamma \ergo \Delta$. Since $\termfal$ has
type $\fal $, the claim holds also for any term occurring in $\Gamma
\ergo \Delta , \termfal : \fal $.


\item $ \vcenter{ \infer[\lo \E]{ \Gamma , \Sigma \ergo ts:B , \Delta
, \Theta }{\Gamma \ergo t: A\lo B, \Delta && \Sigma \ergo s: A,
\Theta} }$. By inductive hypothesis the claim holds for any term
occurring in $ \Gamma \ergo t: A\lo B, \Delta $ and $ \Sigma \ergo s:
A , \Theta$. Since $ts$ is in normal form, also $t$ must be in
normal form. Hence, by inductive hypothesis, $t= y\sigma $. Therefore,
the claim holds for all terms in $\Gamma , \Sigma \ergo ts:B , \Delta
, \Theta $ as well.

\item $ \vcenter{ \infer[\lo \I]{ \Gamma \ergo \send{x} t : A \lo B ,
\Delta }{\Gamma , x: A \ergo t: B , \Delta }}$. By inductive
hypothesis the claim holds for any term occurring in $\Gamma , A \ergo
B , \Delta$. Since $\send{x} t : A \lo B$ is a value, the claim holds
for all terms in $ \Gamma \ergo \send{x} t : A \lo B , \Delta $ as
well.

%
 
\item $ \vcenter{ \infer[\pa \E]{ \Gamma ,\Sigma_1 ,\Sigma_2 \ergo \distp{x}{y} t:
\fal , \Delta , \Theta_1 , \Theta_2 }{\Gamma \ergo t: A\pa B , \Delta
&& \Sigma _1 , x: A \ergo \Theta _1 && \Sigma _2 , y: B \ergo  \Theta _2 }} $.\\By inductive hypothesis the claim holds for any
term occurring in $\Gamma \ergo t: A\pa B , \Delta $ and $ \Sigma _1
, x: A \ergo \Theta _1 $ and $ \Sigma _2 , y: B \ergo \Theta _2$.
Since $\distp{x}{y} t: \fal$ has type $\fal$, the claim holds for all
terms in $ \Gamma ,\Sigma_1 ,\Sigma_2 \ergo \distp{x}{y} t: \fal ,
\Delta , \Theta_1 , \Theta_2$ as well.

\item $ \vcenter{ \infer[\pa \I ]{ \Gamma \ergo t\p s : A \pa B ,
\Delta}{\Gamma \ergo t: A , s: B, \Delta}}$. By inductive hypothesis
the claim holds for any term occurring in $\Gamma \ergo t: A , s: B,
\Delta$. Since $t\p s : A \pa B$ is a value, the claim holds for
$\Gamma \ergo t\p s : A \pa B , \Delta$ as well.
\end{itemize}
\end{proof}

 We can now show that normal
$\lamp$-terms satisfy the Subformula Property. The proof is by induction on the size of the $\NLL$
derivation and, as usual, the difficult case is the one involving
implication elimination. Even though we handle this case by a standard
argument on the type of sequences of eliminations, represented here by
stacks, the multiple-conclusion setting forces us to integrate this
argument in the main induction. We do so by using a stronger inductive
statement that carries along the induction the required statement
about the type of stacks.
\begin{theorem}[Subformula Property] \label{proof:subform}
Consider any normal $\NLL$ derivation $\calp$ of the sequent $x_1 :X_1 , \ldots ,x_m:X_m \ergo t_1: T_1 , \ldots ,
t_n:T_n, \Pi$, where $\Pi$ does not contain any type.
Then every type $ S$ occurring in $\calp$
 is a subformula of some type $T_1 , \ldots , T_n $ or $ X_1
, \ldots , X_m $ or $ \fal$.

\end{theorem}
\begin{proof} We prove a stronger statement:

\begin{quote} Consider any normal $\NLL$ derivation $\calp$ of the
sequent $x_1 :X_1 , \ldots ,x_m:X_m \ergo t_1: T_1 , \ldots ,
t_n:T_n, \Pi$, where $\Pi$ does not contain any type.
 Then every type $ S$ occurring in $\calp$
 is a subformula of some type $T_1 , \ldots , T_n $ or $ X_1
, \ldots , X_m $ or $ \fal$. Moreover, if a term $t_i$ with $i\in
\{1 , \ldots , n\}$ is of the form $y\sigma$ where $\sigma$ is a stack, then $T_i$ is a subformula of some type $T_1 ,
\ldots , T_{i-1} , T_{i+1}, \ldots , T_n $ or $ X_1 , \ldots , X_m $
or $ \fal$.
\end{quote}

We proceed by induction on the number of rule applications in $\calp$.

If no rule is applied in the derivation, then the root of $\calp$ is
$x: A \ergo x: A$  and the claim holds.

Suppose now that the statement holds for any $\NLL$ derivation
containing $q$ or less rule applications, we show the result for a
generic $\NLL$ derivation containing $q+1$ rule applications. We
reason by cases on the last rule applied in this $\NLL$ derivation.
\begin{itemize}

\item $\vcenter{\infer[\fal\I]{\Gamma \ergo \Delta , \termfal : \fal
}{\Gamma \ergo \Delta}}$. By inductive hypothesis the claim holds for
the derivation of the sequent $ \Gamma \ergo \Delta$, which must be normal, thus it holds also for $\calp$.

\item $\vcenter{\infer[\fal \E]{\Gamma \ergo \Delta, \abort(t)
}{\Gamma \ergo \Delta, t: \fal}}$. By inductive hypothesis the claim holds for
the derivation of the sequent $ \Gamma \ergo \Delta, t: \fal$, which must be normal, thus it holds also for $\calp$.


\item $ \vcenter{ \infer[\lo \E]{ \Gamma , \Sigma \ergo ts:B , \Delta
, \Theta }{\Gamma \ergo t: A\lo B, \Delta && \Sigma \ergo s: A,
\Theta} }$. By inductive hypothesis the claim holds for the
derivations of the sequents $ \Gamma \ergo t: A\lo B, \Delta $ and $
\Sigma \ergo s: A , \Theta$, which must be both normal. Let $S$ be any type occurring in $\calp$. If $S$ is
a subformula of some type in $\Gamma , \Sigma, B , \Delta, \Theta$, then the claim holds. If $S$, on the other hand, is
a subformula of $A\lo B$ or $A$, we consider the form of the term $t: A\lo B$.  Since $t: A\lo B$, if $t$ were a
value, then by
Proposition~\ref{prop:coherence},  $t=\send{x} u$. Since the premises of the inference do not share variables,  $x$ does not occur in $s$,  thus it occurs in $u$ or $\Delta, \Theta$ and it would be possible to apply a reduction rule and $\calp$ would not be in normal form. Therefore,
$t$ cannot be a value. Hence, by Lemma~\ref{lem:stackform}, $t=
y\sigma$ where $\sigma$ is a stack. Thus, by inductive hypothesis,
$A\lo B $ is a subformula of some type in $\Gamma $ or $\Delta $  and so is $S$.
 Moreover, the type $B$ of $ts$ is a subformula
of some type in $\Gamma $ or $\Delta $. Hence, the claim holds for
$\calp$ as well.

\item $ \vcenter{ \infer[\lo \I]{ \Gamma \ergo \send{x} t : A \lo B ,
\Delta }{\Gamma , x: A \ergo t: B , \Delta }}$. By inductive
hypothesis the claim holds for the derivation of the sequent $\Gamma ,
A \ergo B , \Delta$, which must be normal. Since all types which are subformulae of $A$ or
$B$ are subformulae of $A\lo B$ as well and since we do not have to
prove anything more about the term $\send{x} t : A \lo B$, the claim
holds for $\calp$ as well.

\item $ \vcenter{\infer[\pa \E]{ \Gamma ,\Sigma_1 ,\Sigma_2 \ergo
\distp{x}{y} t: \fal , \Delta , \Theta_1 , \Theta_2 }{\Gamma \ergo t:
A\pa B , \Delta && \Sigma _1 , x: A \ergo \Theta _1 && \Sigma _2
, y: B \ergo  \Theta _2 }}$.  \\By inductive hypothesis the claim
holds for the derivations of the sequents $\Gamma \ergo t: A\pa B ,
\Delta $ and $ \Sigma _1 , x: A \ergo \Theta _1 $ and $ \Sigma _2 , y:
B \ergo \Theta _2$, which must be all normal.  Let $S$ be any type occurring in $\calp$. If $S$ is
a subformula of some type in $\Gamma , \Sigma, \fal , \Delta, \Theta$, then the claim holds. If $S$, on the other hand, is
a subformula of 
$A\pa B $, we consider the form of the term $t: A\pa B$. Since $t: A\pa B$, if $t$ were a
value, then by
Proposition~\ref{prop:coherence},  $t=u\p v$. Since $x$ and $y$ must occur in $\Theta$, it would be possible to apply a reduction rule and $\calp$ would not be in normal form. Therefore,
$t$ cannot be a value. Hence, by
Lemma~\ref{lem:stackform}, $t= y\sigma$ where $\sigma$ is a
stack. Thus, by inductive hypothesis, $A\pa B $ is a subformula of
some type in $\Gamma $ or $\Delta $ and thus the claim holds with
respect to $\calp$ for all terms in the derivation the type of which
is a subformula of $A\pa B$. Since we do not have to prove anything
more about the term $\distp{x}{y} t: \fal$, the claim holds for
$\calp$ as well.

\item $ \vcenter{ \infer[\pa \I ]{ \Gamma \ergo t\p s : A \pa B ,
\Delta}{\Gamma \ergo t: A , s: B, \Delta}}$. By inductive hypothesis
the claim holds for the derivation of the sequent $\Gamma \ergo t: A ,
s: B, \Delta$ and hence also for $\calp$.
\end{itemize}
\end{proof}

\section{Strong Normalization}\label{sec:normalization}

It is quite immediate to see that all terms of the calculus strongly
normalize. Indeed, each reduction step strictly decreases the size of
the term. This is  due to the linear nature of terms (Proposition \ref{prop:linchan}): each communication or $\lambda$-reduction moves one or two terms from one location to another. Therefore,
no duplication is involved in the reductions. Since, moreover, each
reduction removes one binder from the term, we have that the size of
the term strictly decreases. We formally show this in the following
theorem.

\begin{definition}[Communication-size of a term] For any term $t$, its \emph{communication size} $\cs (t)$
is the number of occurrences of $\send{x}$ and $\distp{x}{y}$ in $t$.
\end{definition}

\begin{theorem}[Strong normalization]\label{thm:sn} For any term $t$ such that the
sequent $\Gamma \ergo t:F$ is derivable in $\NLL$, all sequences of
terms $t_1 , t_2 \ldots$ such that $t_1=t$ and $t_i\red t_{i+1}$
are finite and contain exactly $\cs(t)$ terms.
\end{theorem}
\begin{proof} The proof is by induction on the communication-size
$\cs(t)$ of $t$. In the base case, $t$ does not contain any redex and
the claim trivially holds. We prove then that the claim holds for a
generic term $t$ such that $\cs(t)=n+1 $, under the assumption that
the claim holds for all terms $s$ such that $\cs(u)=n$. In order to do
so, we consider a generic sequence of terms $t_1 , t_2 \ldots$ such
that $t_1=t$ and $t_i\red t_{i+1}$, and we reason by case
distinction on the shape of the reduction $t\red t_2$.
  \begin{itemize}
  \item $t_{1}=s_1 \p \ldots \p \calc [\sendm{x}{v}{w}] \p \ldots \p s_n
\quad \red \quad (\; s_1 \p \ldots \p \calc [v] \p \ldots \p s_n \;
) [w/ x]=t_{2} $.  Since, by Proposition \ref{prop:linchan}, $x$ occurs exactly once in $t_2$ and since $\sendm{x}{v}{w}$ is replaced by $v$, we have that $\cs(t_2)=n$ and that, by
inductive hypothesis, the claim holds for $t_2$. Thus, all sequences
of terms $ t_2, t_3 \ldots$ such that $t_i\red t_{i+1}$ contain
exactly $n$ terms. The sequence $t_1 , t_2 \ldots$ contains thus $n+1$
terms.


  \item {\small $t_{1}=s_1 \p \ldots \p \calc [\distp{x}{y} (v \p w)] \p
\ldots \p s_n \quad \red \quad (\; s_1 \p \ldots \p \calc
[\termfal] \p \ldots \p s_n \; ) [v/x\;\; w/y]=t_{2}$}.  Since, by Proposition \ref{prop:linchan},  both $x$ and
$y$ occur exactly once in $t_2$ and since $\distp{x}{y} ( v\p w)$ is
replaced by $\termfal$, we have that $\cs(t_2)=n$ and that, by
inductive hypothesis, the claim holds for $t_2$. Thus, all sequences
of terms $ t_2, t_3 \ldots$ such that $t_i\red t_{i+1}$ contain
exactly $n$ terms. The sequence $t_1 , t_2 \ldots$ contains thus $n+1$
terms.
\end{itemize}
\end{proof}

\section{Confluence}\label{sec:confluence}

We show now that each $\lamp$-term has a unique normal
form. Since all $\lamp$-terms have a normal form, this implies that
the calculus $\lamp$ is confluent: for any term $t$, if $t\mapsto
t'$ and $t\mapsto t''$, then there is a term $t^\star$ such that both
$t'\mapsto t^\star$ and $t''\mapsto t^\star$.

In order to simplify the proof of confluence, we define the concept of
activator. Intuitively, the activator of a reduction is the term which
is responsible for the reduction. For instance, to trigger a
communication we only need a subterm which is ready to send a message,
such as $\sendm{x}{v}{w}$. If we do trigger this communication of $w$
through $x$, the term $\sendm{x}{v}{w}$ is going to be the activator
of the reduction.
\begin{definition}[Activator] Given any reduction $s\red t$, we say
that the \emph{activator of} $s\red t$ is the subterm of $s$ of the
form $(\lambda x\, v)w$, $\sendm{x}{v}{w}$, $\distp{y}{z} v$,
which is displayed in the corresponding
reduction rule.
\end{definition}

\begin{theorem}[Uniqueness of the normal form]For any term $t$
such that the sequent $\Gamma \ergo t:F$ is derivable in $\NLL$, $t$
has only one normal form.
\end{theorem}
\begin{proof} The proof is by induction on the length of the
normalization of $t$.

If $\cs(t)=0$, the claim trivially holds. We show the claim for a
generic term $t$ such that $\cs(t)=m+1$, under the assumption that the
claim holds for all terms with communication-size $m$ or less.

Suppose that the term $t$ reduces to two different terms $t'$ and
$t''$ if we reduce two different redexes in $t$. We show that both
$t'$ and $t''$ can reduce to the same term $t^\star$. Since, by
inductive hypothesis, $t'$ has a unique normal form and $t''$ has a
unique normal form as well, this is enough to show that $t$ has a
unique normal form.  Since the argument for the $\pa$-reductions
is everywhere close to identical to the argument for  $\lo$-reductions, we only
present  the latter.
Let us denote then the reduction $t\red t'$ as \[u_1 \p \ldots \p \calc
[\sendm{x}{v}{w}] \p \ldots \p u_n \quad \red \quad (\; u_1 \p
\ldots \p \calc [v] \p \ldots \p u_n \; ) [w/ x] \] 
There are two main cases.\\
\textbf{1}. If one of the
activators of $t\red t'$ and $t\red t''$ is a subterm of the
other, we consider the outermost activator. Without loss of
generality, we assume that the outermost is the activator of $t\red
t'$.  Now, the activator $\sendm{y}{s}r$ of $t\red t''$ is either a subterm of $v$
or of $w$.  
\begin{itemize}
\item The activator of $t\red t''$ is a subterm of $v$. The reduction
$t\red t''$ is then of the form\begin{small}\[u_1 \p \ldots \p \calc
[\sendm{x} {v}{w}] \p \ldots \p u_n \quad \red \quad (u_1 \p \ldots \p
\calc [\sendm{x} {v'}{w}] \p \ldots \p u_n)[r/y]\]\end{small}where
$v'$ is obtained from $v$ by replacing $\sendm{y}{s}r$ with $s$.
Since, by  Proposition~\ref{prop-nofreaks}, $y$ does not occur in $w$, we have $t''\red t^\star$, with
 $$t^\star := ((u_1\p \ldots \p \calc [v'] \p \ldots \p u_n)[r/y])[w/
x]$$  We just need to
show that $t'\red t^\star$, that is \begin{small}
    \[ (\; u_1 \p \ldots \p \calc [v] \p \ldots \p u_n \; ) [w/ x]
\quad \red \quad ((u_1\p \ldots \p \calc [v'] \p \ldots \p
u_n)[r/y])[w/ x]\]
  \end{small}
By Proposition~\ref{prop:linchan} we have that $y\neq
x$. Now, if $r$ contains $x$,
\begin{small}
\[ (u_1 \p \ldots \p \calc [v] \p \ldots \p u_n ) [w/ x] = (u_1 \p \ldots \p \calc [v[w/x]] \p \ldots \p u_n )\]
\[ \red
\quad  (\; u_1 \p \ldots \p \calc [v'] \p \ldots \p u_n \; )[r[w/ x]/y]=t^{\star} \]
\end{small}
If $r$ does not contain $x$ and since $y$ does not occur in $w$, we have
\begin{small}
\[ (u_1 \p \ldots \p \calc [v] \p \ldots \p u_n ) [w/ x] \quad \red
\quad  (\; u_1 \p \ldots \p \calc [v'] \p \ldots \p u_n \; )[w/x][r/y]=t^{\star} \]
\end{small}

\item The activator of $t\red t''$ is a subterm of $w$. The reduction
$t\red t''$ is then of the form\begin{small}\[u_1 \p \ldots \p \calc
[\sendm{x}{v}{w}] \p \ldots \p u_n \quad \red \quad (u_1 \p \ldots \p
\calc [(\sendm{x}{v}{w'}] \p \ldots \p u_n)[r/y]\]\end{small}where
$w'$ is obtained from $w$
 by replacing $\sendm{y}{s}r$ with $s$.
  We have $t''\mapsto t^{\star}$, with  $$t^\star:= ((u_1\p \ldots \p \calc [v] \p \ldots \p u_n)[r/y])[w'[r/y]/
x]$$  We just need to
show that $t'\red t^\star$, that is  
\begin{small}
  \[ (\; u_1 \p \ldots \p \calc [v] \p \ldots \p u_n \; ) [w/ x] \quad
    \red \quad   ((u_1\p \ldots \p \calc [v] \p \ldots \p
    u_n)[r/y])[w'[r/y]/ x]\]
\end{small}We have the
following reduction:\begin{small}
  \[ (u_1 \p \ldots \p \calc [v] \p \ldots \p u_n ) [w/ x] \quad \red
\quad ( (\; u_1 \p \ldots \p \calc [v] \p \ldots \p u_n \; ) [w'/
x])[r/y] \]
\end{small}
Now, $r$ does not contain $x$, by Proposition \ref{prop-nofreaks} and because the activator of
$t\red t''$ occurs in $w$. Since, moreover, $y\neq x$, 
 we also have\begin{scriptsize}
  \[ ( (\; u_1 \p \ldots \p \calc [v] \p \ldots \p u_n \; ) [w'/
x])[r/y] \; =\; ( (\; u_1 \p \ldots \p \calc [v] \p \ldots \p u_n \; )
[r/y]) [w'[r/y]/ x]\]
\end{scriptsize}and therefore that $t'\red t^\star$.
\end{itemize}
\textbf{2}. If neither of the activators of $t\red t'$ and $t\red t''$ is a subterm of the other, the reduction
  $t\red t''$ is of the
  form\begin{small}\[u_1 \p \ldots \p \calc [\sendm{x}{v}{w}] \p
      \ldots \p u_n \quad \red \quad (u_1' \p \ldots \p \calc'
      [\sendm{x}{v}{w}] \p \ldots \p u_n')[r/y]\]\end{small}where
  $u_i '$ and  $ \calc'
      [\; ]$ are either equal to, respectively, $u_{i}$ and $\calc[\;]$ or obtained from $u_i $ and  $ \calc
      [\; ]$, respectively, by replacing the activator $\sendm{y}{s}{r}$ of
  $t\red t''$ with $s$. We have $t''\mapsto t^{\star}$, with   
   $$t^\star := ((u_1'\p \ldots \p \calc' [v] \p \ldots \p u_n')[r/y])[w[r/y]/
  x]$$
We just need to
  show that $t'\red t^\star$,   that is 
  \begin{small}
    \[ (\; u_1 \p \ldots \p \calc [v] \p \ldots \p u_n \; ) [w/ x]
      \quad \red \quad ((u_1'\p \ldots \p \calc' [v] \p \ldots \p
      u_n')[r/y])[w[r/y]/ x]\]
  \end{small}Now,
  if $s$ is the activator of a reduction, also $s[w/x]$ has the
  suitable shape to be one; and the relative reductions substitute the
  same variable. By Definition~\ref{def:sequent}, moreover, $y\neq x$
  and hence $y$ occurs in the term
  $(\; u_1 \p \ldots \p \calc [v] \p \ldots \p u_n \; ) [w/ x]$, and
  we can reduce it as follows:
  \begin{small}
    \[ (u_1 \p \ldots \p \calc [v] \p \ldots \p u_n \; ) [w/ x]
      \quad \red \quad ( (\; u_1' \p \ldots \p \calc '[v] \p \ldots
      \p u_n') [w/ x])[r[w/ x]/y] \]
  \end{small}
  We must now consider several cases, according as to whether $x$ and $y$ occurs in the messages $w$ and $r$.\\
   - If $x$ occurs in $r$, but $y$ does not occur in $w$, then, by Proposition \ref{prop:linchan}, $x$ does not occur in $u_1' \p \ldots \p \calc '[v] \p \ldots \p u_n' $ and we obtain indeed
\[ ( (\; u_1' \p \ldots \p \calc '[v] \p \ldots
      \p u_n') [w/ x])[r[w/ x]/y]\] \[= ( \; u_1' \p \ldots \p \calc '[v] \p \ldots
      \p u_n')[r[w/ x]/y]=( \; u_1' \p \ldots \p \calc '[v] \p \ldots
      \p u_n')[r/y][w/ x]=t^{\star}\]
- If $x$ does not occur in $r$, but $y$  occurs in $w$, then, by Proposition \ref{prop:linchan}, $y$ does not occur in $u_1' \p \ldots \p \calc '[v] \p \ldots \p u_n' $ and we obtain indeed
\[ ( (\; u_1' \p \ldots \p \calc '[v] \p \ldots
      \p u_n') [w/ x])[r[w/ x]/y]\] \[= ( \; u_1' \p \ldots \p \calc '[v] \p \ldots
      \p u_n')[w/ x])[r/y]=( \; u_1' \p \ldots \p \calc '[v] \p \ldots
      \p u_n')[w[r/y]/x]=
      t^{\star}\]
      -  If neither $x$ occurs in $r$ nor $y$  occurs in $w$, then, by Proposition \ref{prop:linchan}, $y$ does not occur in $u_1' \p \ldots \p \calc '[v] \p \ldots \p u_n' $ and we obtain indeed
\[ ( (\; u_1' \p \ldots \p \calc '[v] \p \ldots
      \p u_n') [w/ x])[r[w/ x]/y]\] \[= ( \; u_1' \p \ldots \p \calc '[v] \p \ldots
      \p u_n')[w/ x])[r/y]=( \; u_1' \p \ldots \p \calc '[v] \p \ldots
      \p u_n')[r/y][w/x]=
      t^{\star}\]
- Finally, we show it cannot be that both $x$ occurs in $r$ and
$y$ occurs in $w$. Indeed, suppose, for the sake of contradiction, that they do.
By letting the term $\sendm{y}{s}{r}$  transmit its message $r$,  we would have\begin{small}
\[t=u_1 \p  \ldots\p \calc
[\sendm{x}{v}{w}] \p \ldots \p u_n \qquad \red \quad u_1' \p \ldots \p \calc'
[\sendm{x}{v}{w[r/y]}] \p \ldots \p u_n':= t'''  \]
\end{small}
By Theorem \ref{thm-subred}, $\Gamma \ergo t''':F$, which is in contradiction with Proposition \ref{prop-nofreaks}, since $x$ occurs in $w[r/y]$.
\end{proof}

\section{Related Work and Conclusions}\label{sec:related-work}

The origins of $\lamp$ are influenced by and related to several functional and process calculi. We shall consider them one by one.

\subsection*{$\lamp$ and Linear Session Typed $\pi$-calculi}

One of the most established ways to interpret linear logic proofs into
concurrent programs consists in interpreting sequent calculi for
linear logic into $\pi$-calculi with session types,  main references being~\cite{CairesPfenning, Wadler2012, KMP2019}. 
 In those calculi, a session type is a logical expression containing information about the whole sequence
of exchanges that occur between two processes through one
channel. Session types are attached only to  communication channels, therefore they only describe the channels' input/output behaviour, not the processes using them. In $\lamp$, as in the tradition of the Curry-Howard correspondence, types are instead attached to processes, are read as process specifications and are employed to formally guarantee that processes will behave according to the specifications expressed by their type. Nevertheless, since the nature of $\pi$-calculus processes is determined by their channels' behaviour, specifying channels means also specifying processes, hence the difference with $\lamp$ is not as significant as it may appear.
 
 A channel typed
by a session type will occur in general more than once and
will support several transmissions. This is not the case
in $\lamp$. The linearity of the calculus implies that each output channel
can only occur once and thus be used once. As shown in our examples, two processes
of $\lamp$ can nevertheless communicate more than once, since they can
be connected by several different channels. In particular, in $\lamp$
there is no restriction on the number and direction of the
communications between any two processes: dialogues are possible and
information can be exchanged back and forth with no limitation.

Unlike in $\pi$-calculus, there is no need in $\lamp$ of a construct $\nu x$ for declaring a channel private. Since in our calculus any given communication channel connects exactly two processes, the channel is already a private link between the involved processes, simply because there is no other process sharing the same channel. 

The typing system of $\lamp$ supports also cyclically interconnected processes, therefore $\lamp$ cannot be faithfully translated in the linearly session typed  $\pi$-calculi \cite{CairesPfenning, Wadler2012, KMP2019}. 
At the time of this writing, there is no known straightforward translation of the multiplicative fragments of those calculi in $\lamp$, but also no known reason a translation should not be possible.

 \subsection*{$\lamp$ and other Typed Concurrent $\lambda$-calculi}
 
 As a consequence of adopting linear sequent calculus, the session typings of \cite{CairesPfenning, Wadler2012, KMP2019} do not support directly functional computation. Several approaches have been developed to overcome this intrinsic limitation.
 
 The system in \cite{TCP2013} 
 provides some combination of session typed processes and functional programs. However, the typing system for the functional part is just added on top of the linear logic system. Therefore, there is no seamless account of both functional and concurrent computation by a \emph{single system} for linear logic, which is the main accomplishment of our work on $\lamp$. 
 
 The linear concurrent functional calculus GV in \cite{Wadler2012} is more similar in spirit to $\lamp$. However, when it comes to linear logic, it is not as pure as ours and displays no connection between computation and the full normalization process  leading to analytic proofs. No direct reduction rules for GV are provided either.
 
 The two concurrent  $\lambda$-calculi studied in \cite{ACG2017, ACG2018gandalf} are typed by G\"odel and classical logic. They are more expressive than $\lamp$ when it comes to functional computation, because they are based on extensions of intuitionistic natural deduction and thus feature full typed $\lambda$-calculus. However, the reduction rules are significantly more complex than those of $\lamp$, due to the complicated treatment that code mobility requires. Moreover, strong normalization and confluence do not hold. 
 
 The concurrent $\lambda$-calculus of \cite{AschieriZH} is typed by first-order classical logic. It resembles a session-typed $\lambda$-calculus, but supports only data transmission. It is strongly normalizing, but is non-confluent  and does not enjoy any subformula property.
 
 The calculus Lolliproc in \cite{Mazurak} is a parallel interpretation  of control operators,  rather than a calculus centered on communication. The completeness of the typing system with respect to linear logic is not proved and no subformula property is shown either, which makes unclear if it actually is a legitimate proof system for linear logic.

 \subsection*{$\lamp$ and Proof-Nets}
 
Proof-nets  for multiplicative linear logic \cite{Gir87} can encode several programming languages, among them linear $\lambda$-calculus \cite{GLT89} and  typed $\pi$-calculus \cite{Laurent2010}.  We believe $\lamp$ can be encoded in proof-nets as well and hence offers a new way of using proof-nets for representing functional concurrent computation. However, $\lamp$ is based on the connectives $\lo, \pa$ and normalization, while proof-nets are based on the connectives $\te, \pa$ and cut-elimination, which makes $\lamp$  far more suitable as syntax for functional concurrent programming. Indeed, $\lamp$ is a full-blown programming language right off the bat, while proof-nets were created as an elegant and economical proof system for linear logic.

\subsection*{$\lamp$ and Exponentials}
 
As type system of $\lamp$, multiplicative linear logic already features all the important characters  of functional concurrent computation, and because of both its expressiveness and simplicity, it deserves to be treated as a type system on its own, as traditional in linear logic. One non-trivial problem is to extend the type system of $\lamp$ with exponentials. We do not see any real obstacles to doing that, but we leave the problem as future work. As a result, $\lamp$ may gain more duplication and replication abilities, as those of CP \cite{Wadler2012}.

\subsection*{$\lammu$ and $\lamp$: Sequentiality vs Parallelism}

As far as the logical rules for
implication are concerned, the type system of the
original version of $\lammu$~\cite{Par91} is exactly the non-linear version of the
type system of $\lamp$. However, $\lammu$ and $\lamp$ strongly diverge
if we consider their computational behaviour. In
particular, the crucial difference between $\lamp$ and $\lammu$ is
that while in $\lammu$ a whole list of formulas is used to type a
single term: $t: A_1, \ldots , A_n$, in $\lamp$ each element of the list of formulas
types a different term: $t_1: A_1, \ldots , t_n:A_n$. Indeed, although the type system of $\lammu$ actually builds several different terms, it does it  sequentially, so the result is a unique term; instead, the type system of $\lamp$ builds several different terms in parallel and let them communicate. 

The reason is that
$\lammu$ forces only one distinguished ``active'' formula of each list to be used
at any time in the derivation. 
Labels $\alpha, \beta \dots $ -- also called $\mu$-variables -- are
introduced to explicitly activate a formula, rule below left, which
can therefore be used in an inference, or
deactivate it, rule below right:
\[ \infer{ \mu \beta\, t : \Gamma \ergo A, \Delta}{t: \Gamma \ergo
A^{\beta}, \Delta}\qquad \qquad \qquad \qquad \infer{[\beta]t: \Gamma
\ergo A^{\beta}, \Delta }{t: \Gamma \ergo A, \Delta} \]where $\Gamma$
is a context containing declaration of $\lambda$-variables $x,
y, \dots $ and $\Delta$ is a multiset containing formulae labeled by
$\mu$-variables $\alpha, \beta \dots$. 
As we can see from the condition on the context $\Delta$,
at most one formula in each sequent can be without  a label.


The history of activations and
deactivations occurring in the derivation is recorded inside the corresponding $\lammu$-term, as we can see in the following example
\begin{small}
\[\infer{(\mu \beta . u) \, v : {B}, \Delta , \Sigma }{\infer[]{\mu
\beta .u: {A\impl B}, \Delta }{\infer*{u: {(A\impl B)^{\beta}},
\Delta}{\infer[]{[\beta ]w : {(A \impl B)^{\beta}}, \Delta
'}{\infer*{w : {A \impl B}, \Delta '}{}}}} & \infer*{v:{A},
\Sigma}{}}\] 
\end{small}
Here the construction of $w$ has been paused for a while and resumed with the last inference rule. Hence, $v$ can be considered as the ``real'' argument of $w$, rather than of $ (\mu \beta . u )v$. As consequence,
there is the reduction $ (\mu \beta . u )v\;
\mapsto\; \mu \beta . u[[\beta](wv)/[\beta]w] $, where the argument $v$ of
$\mu \beta. u$ is transmitted as argument to all subterms of $t$ with label
$[\beta]$.  

\subsection*{Multiple-conclusion Full Intuitionistic Linear Logic}

The $\lambda$-calculus in \cite{paivalinear} provides a computational interpretation of an intuitionistic linear logic with the $\pa$ connective. The typing system however is a linear sequent calculus, so the corresponding $\lambda$-calculus is rather obtained by translation  than by a Curry-Howard isomorphism between proofs and $\lambda$-terms. To obtain such an isomorphism, our move to natural deduction was necessary. Also, in \cite{paivalinear} cut-elimination is not interpreted as usual as communication, so the system does not support concurrency. A radical change in the treatment of linear implication was needed to recover communication and an interpretation of $\pa$ as parallel operator was needed to recover parallelism, which is what we did. It is also instructive to compare  $\pa\E$ rule to the corresponding rule in~\cite{paivalinear}:
\[\infer[\pa \E_2 ]{\Gamma , \Sigma_1 ,\Sigma _2 \ergo \mathtt{let} \,
s \,\mathtt{be} (A \pa - ) \mathtt{in} \, v : C , \mathtt{let} \, s
\,\mathtt{be} ( - \pa B ) \mathtt{in}\, w : D , \Delta , \Theta _1
,\Theta _1 }{ \Gamma \ergo s: A \pa B , \Delta && \Sigma_1 , x: A
\ergo v:C, \Theta _1 && \Sigma _2 , y:B \ergo w: D, \Theta _2 }\]
This rule  duplicates the term $s$, compromising linearity of variables. 
Nevertheless, 
the linearity of the computation is  safe if we
consider that $s$ virtually occurs twice, but it will only be used
once: one half of it, inside $v$, the other half, inside $w$; the two
spare halves will be discarded.
Instead,  $\lamp$ is built around communication. We thus re-interpreted
disjunction redexes as communications in order to avoid unnecessary duplication in case
distinction constructs. Indeed, what we need in interpreting a $\pa \E
$ rule with premises $\Gamma \ergo s: A \pa B , \Delta$, $\quad
\Sigma_1 , x: A \ergo v:C, \Theta _1\quad $ and $ \quad \Sigma _2 ,
y:B \ergo w: D, \Theta _2 $ is simply a way to provide one half of $s$
to $v$ and the other half to $w$ when $s$ will assume a form which is
suitable for a reduction. Instead of duplicating $s$ and attaching one
copy of it to $v$ and one to $w$, we  just let $s$ be an autonomous
term -- living in parallel with all other terms in our multiple-conclusion -- and established a communication channel from $s$ to both
$v$ and $w$. Thus, when $s$ assumes the right form, we triggered a
communication transmitting part of $s$ to $v$ and part to $w$ as
required.

\bibliographystyle{plain}

\bibliography{bib_clinear}

\end{document}